\documentclass[
notitlepage
,floatfix,
aps,
pra,
reprint,
superscriptaddress,
longbibliography,
11pt
]{revtex4-2}
\usepackage[utf8]{inputenc}

\usepackage[caption=false]{subfig}
\usepackage{graphicx}
\usepackage{epstopdf}
\usepackage{array}
\usepackage{verbatim}
\usepackage{amsmath,amsfonts,amssymb,amscd,mathtools}
\usepackage{dsfont}
\usepackage{amsthm}
\usepackage{tabularx}
\usepackage{stmaryrd}
\usepackage{enumerate}
\usepackage{booktabs}
\usepackage{physics}

\usepackage{wasysym}
\usepackage{mathrsfs}
\usepackage{microtype}
\usepackage{diagbox}
\usepackage{tikz}
\usetikzlibrary{positioning}

\usepackage[lmargin=.7in,rmargin=.7in,tmargin=.7in,bmargin=1in]{geometry}

\usepackage{newtxtext,newtxmath}

\usepackage{hyperref}
\usepackage[dvipsnames]{xcolor}
\hypersetup{
    bookmarksnumbered=true, 
    unicode=false, 
    pdfstartview={FitH}, 
    pdftitle={}, 
    pdfauthor={}, 
    pdfsubject={}, 
    pdfcreator={}, 
    pdfproducer={}, 
    pdfkeywords={}, 
    pdfnewwindow=true, 
    colorlinks=true, 
    linkcolor=NavyBlue, 
    citecolor=NavyBlue, 
    filecolor=NavyBlue, 
    urlcolor=NavyBlue 
}
\usepackage{aliascnt}
\theoremstyle{plain}
\newtheorem{thm}{Theorem}
\newaliascnt{cor}{thm}
\newtheorem{cor}[cor]{Corollary}
\aliascntresetthe{cor}
\newaliascnt{lem}{thm}
\newtheorem{lem}[lem]{Lemma}
\aliascntresetthe{lem}
\newaliascnt{pro}{thm}
\newtheorem{pro}[pro]{Proposition}
\aliascntresetthe{pro}
\theoremstyle{definition}
\newaliascnt{defn}{thm}
\newtheorem{defn}[defn]{Definition}
\aliascntresetthe{defn}
\newaliascnt{remark}{thm}

\aliascntresetthe{remark}
\newaliascnt{ex}{thm}

\aliascntresetthe{ex}
\newaliascnt{question}{thm}

\aliascntresetthe{question}
\newaliascnt{problem}{thm}

\aliascntresetthe{problem}

\usepackage{cleveref}
\crefname{thm}{Theorem}{Theorems}
\crefname{cor}{Corollary}{Corollaries}
\crefname{lem}{Lemma}{Lemmas}
\crefname{pro}{Proposition}{Propositions}
\crefname{defn}{Definition}{Definitions}
\crefname{remark}{Remark}{Remarks}
\crefname{ex}{Example}{Examples}
\crefname{question}{Question}{Questions}
\crefname{problem}{Problem}{Problems}
\Crefname{equation}{Equation}{Equations}
\crefname{figure}{Fig.}{Figs.}
\Crefname{figure}{Figure}{Figures}
\crefname{equation}{Eq.}{Eqs.}
\Crefname{equation}{Equation}{Equations}
\crefname{figure}{Fig.}{Figs.}
\Crefname{figure}{Figure}{Figures}
\crefname{section}{Sec.}{Secs.}
\Crefname{section}{Section}{Sections}

\usepackage{algorithm}
\usepackage{algorithmicx}
\usepackage{algpseudocode}



\newcommand{\mN}{\mathcal{N}}
\newcommand{\mU}{\mathcal{U}}
\newcommand{\mT}{\mathcal{T}}
\newcommand{\mD}{\mathcal{D}}

\newcommand{\mH}{\mathcal{H}}
\newcommand{\mM}{\mathcal{M}}

\newcommand{\mB}{\mathcal{B}}

\newcommand{\mP}{\mathcal{P}}

\newcommand{\mS}{\mathcal{S}}

\newcommand{\mbD}{\mathbb{D}}

\newcommand{\mbR}{\mathbb{R}}

\newcommand{\mbZ}{\mathbb{Z}}
\newcommand{\mdI}{\mathds{1}}

\newcommand{\ceil}[1]{{\lceil #1 \rceil}}
\newcommand{\floor}[1]{{\left\lfloor #1 \right\rfloor}}

\newcommand{\ba}{\begin{eqnarray}}
\newcommand{\ea}{\end{eqnarray}}
\newcommand{\bann}{\begin{eqnarray*}}
\newcommand{\eann}{\end{eqnarray*}}
\newcommand{\bal}{\begin{equation}\begin{aligned}}
\newcommand{\eal}{\end{aligned}\end{equation}}
\newcommand{\dm}[1]{\ketbra{#1}{#1}}

\newcommand{\sbar}{\;\rule{0pt}{9.5pt}\right|\;}
\newcommand{\lset}{\left\{\left.}
\newcommand{\rset}{\right\}}

\DeclareMathOperator{\poly}{poly}

\DeclareMathOperator{\id}{id}

\usepackage{titlesec}
\titleformat{\paragraph}[block]
  {\normalfont\normalsize\bfseries}
  {\theparagraph}{1em}{}

\titlespacing*{\paragraph}
  {0pt}{3.25ex plus 1ex minus .2ex}{1ex}

\newcounter{protocolstep}[subsubsection]
\renewcommand{\theprotocolstep}{Step~\arabic{protocolstep}}

\newcommand{\stepparagraph}[1]{%
    \refstepcounter{protocolstep}%
    \paragraph*{\theprotocolstep\ of the protocol: #1}%
}

\begin{document}

\title{Universal distillation of quantum entanglement}

\author{Ryuji Takagi}
\email{ryujitakagi@g.ecc.u-tokyo.ac.jp}
\affiliation{Department of Basic Science, The University of Tokyo, 3-8-1 Komaba, Meguro-ku, Tokyo 153-8902, Japan}

\author{Kaito Watanabe}
\email{watanabe715@g.ecc.u-tokyo.ac.jp}
\affiliation{Department of Basic Science, The University of Tokyo, 3-8-1 Komaba, Meguro-ku, Tokyo 153-8902, Japan}
\affiliation{RIKEN Center for Quantum Computing (RQC), Hirosawa 2-1, Wako, Saitama 351-0198, Japan}

\author{Takaya Matsuura}
\email{takayamatsuura@gmail.com}
\affiliation{RIKEN Center for Quantum Computing (RQC), Hirosawa 2-1, Wako, Saitama 351-0198, Japan}

\author{Hayato Arai}
\email{h.arai6626@gmail.com}
\affiliation{Department of Basic Science, The University of Tokyo, 3-8-1 Komaba, Meguro-ku, Tokyo 153-8902, Japan}

\author{Masahito Hayashi}
\email{hmasahito@cuhk.edu.cn}
\affiliation{School of Data Science, The Chinese University of Hong Kong, Shenzhen, Guangdong, 518172, China}
\affiliation{International Quantum Academy, Futian District, Shenzhen 518048, China}
\affiliation{Graduate School of Mathematics, Nagoya University, Chikusa-ku, Nagoya 464–8602, Japan}

\begin{abstract}
Distillable entanglement---the maximum rate of Bell states extractable from the given noisy entangled states via local operations and classical communication (LOCC)---has been a central quantity to characterize the operational value of entangled states. However, this operational meaning had been known to be valid only under a crucial assumption: both parties know the description of the input state and can choose the protocol tailored to it, which does not necessarily reflect the natural operational setting. Here, we show that the distillable entanglement is the fundamental operational quantity even without this assumption. We present a universal LOCC protocol---whose description does not depend on the input state and thus works universally for any input---that achieves the distillable entanglement of an arbitrary unknown input state. Our result therefore lifts the distillable entanglement to a robust notion of an operationally meaningful quantity and clarifies that input-state knowledge is not necessary to achieve the optimal rate of entanglement distillation. Our protocol is based on state merging for compound states, together with the uniform convergence property of the finite-block quantity toward distillable entanglement.   
\end{abstract}

\maketitle

\let\oldaddcontentsline\addcontentsline
\renewcommand{\addcontentsline}[3]{}

\section{Introduction}

Entanglement distillation has been a central problem studied intensively in quantum information theory~\cite{bennett_1996-5,Bennett1996concentrating,Horodecki-review,Devetak2005distillation,hayashi2016mathematical}. 
Its importance can be directly seen by considering that, in any practical operational setting, one cannot avoid noise from occurring in the entangled state in hand. 
To run quantum information processing protocols reliably, it is crucial to prepare pure entangled states from the given noisy entangled states. 
The aim of entanglement distillation is precisely to accomplish this using local operations and classical communication (LOCC)---a natural class of operations accessible to two parties that are physically separated apart. 
From this operational perspective, the key quantity to study is the \emph{distillable entanglement}, defined as the maximum rate of the standard Bell state extractable from the given entangled state by LOCC~\cite{Horodecki-review,hayashi2016mathematical}. 
Although it is notoriously hard to compute in general, it serves as one of the most fundamental entanglement measures that characterizes the operational usefulness of the generic bipartite entangled state. 

However, the operational significance of distillable entanglement is ensured only in a setting with an implicit but crucial assumption---both parties know the complete description of the given state beforehand, allowing them to design distillation protocols specifically tailored to the given state. 
In many cases, this is not a natural assumption because possessing quantum states in the lab---e.g., as an output of quantum computation---does not mean that they have access to the classical description of the state. 
This motivates us to consider \emph{state-agnostic} entanglement distillation without knowing the input state, which can be realized by a protocol that does not depend on the input state. 
This fixed protocol should work \emph{universally} for arbitrary input states and produce the output maximally entangled state, whose size depends on the input state. 
Characterizing the performance of the universal protocol not only removes the unsatisfactory assumption but also provides a deeper understanding of the fundamental operational role played by the information of the input state. 

By definition, the optimal rate of the universal protocol is no larger than the standard distillable entanglement formalized under the state-aware setting. 
Therefore, the best one could hope for is to achieve the distillable entanglement for every input state by a fixed universal protocol, and whether it is possible has been an outstanding open problem~\cite{Horodecki2003entanglement}. 
Along this line, a significant first step was initiated by Ref.~\cite{Matsumoto2007universal}, which solved the case of pure input states, where they showed that the optimal rate characterized by the von Neumann entropy of entanglement can also be achieved universally.
(See also Ref.~\cite{Blume-Kohout2014} for its sequential implementation.)
On the other hand, the understanding of the general mixed-state input is largely limited. 
For the case of unknown Bell diagonal states, a universal version of hashing and recurrence was constructed~\cite{Brun2001entanglement}, and other subclasses of bipartite qubit states have also been discussed~\cite{Horodecki2003entanglement,Kalman2025universal}.
When the given state is promised to belong to the known set of states, it was shown that one can construct a protocol depending on the description of the set---but not on the input state itself---that achieves the minimum coherent information over the set~\cite{Bjelakovic2013universal}. 
Recently, in the relaxed setting where available operations are entanglement non-generating operations~\cite{Brandao-Datta}---which construct a set strictly larger than LOCC---it has been shown that the state-aware optimal rate characterized by the regularized relative entropy of entanglement~\cite{Brandao2010,hayashi_generalized_2025,Lami_2025_gqsl} can also be achieved universally~\cite{lami2026universalquantumresourcedistillation}. 
Although these developments provide significant insights, they only address special cases or different operational settings, and the original problem had still remained unsolved.

Here, we resolve this problem, showing that distillable entanglement can be achieved universally by a state-agnostic LOCC protocol for arbitrary unknown input states.
Our results therefore establish distillable entanglement as a robust operationally meaningful quantity in a much broader setting than it was originally formulated.
To this end, we begin by extending the result of Refs.~\cite{Bjelakovic2013universal,Boche2014resource} and obtain a universal distillation protocol achieving the hashing bound~\cite{Devetak2005distillation,Horodecki2007quantum}---a known lower bound for the distillable entanglement. 
This not only provides an important subroutine to achieve our main goal but also offers a simple and concrete protocol. 
We then extend this to the protocol achieving the distillable entanglement. 
Our approach is to focus on one formal expression of distillable entanglement represented as a limit of a series of functions characterized by coherent information. 
We find that the main subtlety in constructing the universal distillation protocol lies in the mathematical property of this convergent series, particularly its uniform convergence. 
We show the uniform convergence for an arbitrary compact subset of full-rank distillable states and devise the protocol that employs this property.  
We further apply our construction to obtain a general lower bound on the quantum capacity of unknown channels with LOCC assistance, which becomes optimal for teleportation covariant channels~\cite{pirandola_2017,Kaur2017}.   

Beyond the contributions to entanglement theory, our results can be positioned in a rich landscape of relevant subfields. 
The idea of universal protocols has been discussed in various forms in the context of Shannon theory, on topics such as source and channel coding~\cite{Wolfowitz1960Simultaneous,Blackwell1959,Lempel-Ziv,ZivLempel1978,Csiszar2011information,Josza1998universal,Josza2003,Hayashi2002quantum,Bennett2006universal,Hayashi2002simpleconstructionquantumuniversal,Bjelakovic2008quantum,Bjelakovic2009,Bjelakovic2009Classical,Berta_2017_compound_entanglement_assisted,Boche_entanglement-assisted_2017,Hayashi2010,hayashi_2009-2,Hayashi_2022_universal_cq_superposition,matsuura2025universalclassicalquantumchannelresolvability,watanabe2026needuniversalcorrelationdetector}, and hypothesis testing~\cite{brandao_adversarial,Brandao2010,berta_composite,Bjelakovic2005,Noetzel2014,hayashi_generalized_2025,Lami_2025_gqsl,Fang2026,Fang_2026_error_exponent,lami2026universalquantumresourcedistillation,Dasgupta2025}. 
The extension of the notion of universal protocols to quantum resource distillation has seen recent developments, mainly in the general resource-theoretic setting with axiomatic free operations~\cite{Fang2026,lami2026universalquantumresourcedistillation,Lin2026FiniteBlocklength}, as well as in specific resource theories equipped with operationally motivated free operations~\cite{watanabe_2024,Watanabe_universal,faist2026universalthermodynamicimplementationprocess}.
Our work adds an item to the last category but is inspired by the Shannon-theoretic communication task, suggesting a potential for leveraging and extending information-theoretic techniques to obtain further insights into the ultimate capability of universal resource manipulation.

\section{Preliminaries}
\subsection{Entanglement distillation}

Here, we review the standard setting of entanglement distillation from known i.i.d. states. 
Let $\mD(\mH)$ denote the set of all states acting on the Hilbert space $\mH$.  
We consider the setting where two parties, Alice and Bob, possess a bipartite state $\rho_{AB}\in\mD(\mH_{AB})$, where $\mH_{AB}=\mH_A\otimes \mH_B$ with $\mH_A$ and $\mH_B$ being finite-dimensional Hilbert spaces with dimensions $d_A$ and $d_B$. 
Entanglement distillation aims to prepare as many maximally entangled state $\Phi_{AB}$ as possible using LOCC.
The distillable entanglement $E_d(\rho_{AB})$ of the state $\rho_{AB}$ is particularly defined by the maximum rate $r$ such that there exists a series $\{\Lambda_n\}_n$ of LOCC operations such that $\Lambda_n(\rho_{AB}^{\otimes n})$ approaches $\Phi_{AB}^{\otimes \floor{rn}}$ in the asymptotic limit $n\to\infty$.

As can be seen in the definition, the distillable entanglement is given as a result of the optimization over all possible LOCC protocols \emph{for the given $\rho_{AB}$}. 
This particularly means that the optimal series $\{\Lambda_n\}_n$ of LOCC protocols can generally depend on the description of $\rho_{AB}$.
In other words, this is the ultimate performance under the assumption that the experimenter has complete knowledge about the description of $\rho_{AB}$.

The distillable entanglement is a fundamental quantifier of entanglement that directly reflects the operational significance of the entanglement contained in the given state. 
On the other hand, its closed form is not known, and computing $E_d$ is generally a hard problem.
This also motivates us to obtain a more tractable achievable lower bound that applies to an arbitrary state $\rho_{AB}$. 
The standard tractable lower bound is the hashing bound~\cite{Devetak2005distillation,Horodecki2007quantum}, giving 
\bal
 E_d(\rho_{AB})\geq I_c(A|B)_{\rho_{AB}}
 \label{eq:hashing bound main}
\eal
where $I_c(A|B)_{\rho_{AB}} \coloneqq H(\rho_B)-H(\rho_{AB})$ is the coherent information, and $H(\rho)\coloneqq-\Tr(\rho\log\rho)$ is the von-Neumann entropy of the state $\rho$.
In fact, the hashing bound can be achieved by one-way LOCC, an LOCC operation that involves only classical communication from Alice to Bob. 
One of the specific protocols achieving this bound plays a central role in our result, and it is based on the task known as state merging, which we review in the following.

\subsection{State merging}

State merging is a quantum information processing task where Alice aims to transfer her part of the known given state to Bob while maintaining the correlation with the environment by LOCC and entanglement~\cite{Horodecki2007quantum,Horodecki2005a}.
More specifically, they are given the initial state $\rho_{AB}^{\otimes n}$ with purification $\psi_{ABE}^{\otimes n}$.
They also begin with the initial entanglement $\Phi_{d_{\rm ini}}$ and end up with the final entanglement $\Phi_{d_{\rm out}}$.
Then, if there exists a sequence $\{\Lambda_n\}_n$ of LOCC operations such that
\bal
 \norm{\Lambda_n(\psi_{ABE}^{\otimes n}\otimes\Phi_{d_{\rm ini}^{(n)}})-\psi_{\tilde B BE}^{\otimes n}\otimes \Phi_{d_{\rm out}^{(n)}}}_1\xrightarrow[n\to\infty]{} 0
\eal
where $\norm{\cdot}_1$ is the trace norm, we say that $\limsup_{n\to \infty}\frac{\log d_{\rm ini}^{(n)}-\log d_{\rm fin}^{(n)}}{n}$ is an achievable entanglement cost for state merging. 

It is known that the optimal entanglement cost is given by the conditional entropy $H(A|B)_{\rho_{AB}} = H(\rho_{AB})-H(\rho_B)$, which can be achieved by a measurement on Alice's side, followed by Bob's conditional operation.
Interestingly, this also encompasses the case where the conditional entropy is negative.
In this case, we can rather extract entanglement with the rate of coherent information $I_c(A|B)_{\rho_{AB}}=-H(A|B)_{\rho_{AB}}$.
Moreover, a close look at the construction in Ref.~\cite{Horodecki2007quantum} reveals that when $I_c(A|B)_{\rho_{AB}}>0$, one can start with no initial entanglement, i.e., $d_{\rm ini}=1$.
This shows that, as a ``byproduct'' of the state merging task, one can realize entanglement distillation with input state $\rho_{AB}$ without initial entanglement with the rate $I_c(A|B)_{\rho_{AB}}$ of distilled entanglement.

Crucially, the result in Ref.~\cite{Horodecki2007quantum} assumes that Alice and Bob know the description of the state beforehand.
The relevant setting of our interest is the case when the description of $\rho_{AB}$ is not known.
An extension to the setting where the full description of $\rho_{AB}$ is not available was considered in Refs.~\cite{Bjelakovic2013universal,Boche2014resource,Colomer2024}. 
They considered the setting where the state $\rho_{AB}$ is taken from the known subset $\mS$ of quantum states.
Namely, the experimenter has access to \emph{partial} information about the state $\rho_{AB}$ in the sense that they know the base set to which the given state belongs.
In this setting, they investigated the required entanglement cost so that there exists a fixed protocol that accomplishes the merging task for every state in $\mS$.
Their result shows that the entanglement cost is characterized by the maximum conditional entropy of the state in the set.

\begin{lem}[State merging with partial information~\cite{Bjelakovic2013universal}]\label{lem:universal state merging main}
    Let $\mS$ be an arbitrary set of bipartite states on the given system $AB$. Then, there exists a series $\{\Lambda_n\}_n$ of one-way LOCC protocol from Alice to Bob and $\{d_{\rm ini}^{(n)}\}_n$, $\{d_{\rm out}^{(n)}\}_n$ such that, for every input state $\rho_{AB}$ in $\mS$ with purification $\psi_{ABE}$, it holds that
    \bal
     \norm{\Lambda_n(\psi_{ABE}^{\otimes n}\otimes\Phi_{d_{\rm ini}})-\psi_{\tilde B BE}^{\otimes n}\otimes \Phi_{d_{\rm out}}}_1\xrightarrow[n\to\infty]{}0
     \label{eq:state merging partial info fidelity convergence main}
    \eal
    with
    \bal
     \frac{\log d_{\rm ini}^{(n)}-\log d_{\rm out}^{(n)}}{n} \xrightarrow[n\to\infty]{} \sup_{\rho_{AB}\in\mS} H(A|B)_{\rho_{AB}}.
    \eal
When $\sup_{\rho_{AB}\in\mS}H(A|B)_{\rho_{AB}}<0$, no initial entanglement is required, i.e., $d_{\rm ini}^{(n)}=1$ suffices.
\end{lem}

\cref{lem:universal state merging main} particularly means that entanglement distillation from an unknown state from the set $\mS$ can be accomplished with the rate $\min_{\rho_{AB}\in\mS}I_c(A|B)_{\rho_{AB}}$.

\section{Universal entanglement distillation}

\subsection{Setting}

We now formalize the setting of universal entanglement distillation. 
Suppose that Alice and Bob are given i.i.d. copies of the state $\rho_{AB}$, but they do not know its description.
All they can do in this setting is to follow some strategy that does not depend on $\rho_{AB}$, aiming to obtain as many entanglement bits as possible.

More formally, for a given i.i.d. series $\qty{\rho_{AB}^{\otimes n}}_n$ of the unknown input state $\rho_{AB}$, they apply some series $\{\Lambda_n\}_n$ of LOCC maps. 
Here, we are interested in the protocol such that, after the distillation process, Alice and Bob get to know how much entanglement has been extracted. 
This necessarily involves measurement, integrated into the LOCC strategy. 
This means that each LOCC protocol $\Lambda_n$ is written as an instrument $\Lambda_n = \sum_x \Lambda_n^x$ where each $\Lambda_n^x$ is a completely positive trace nonincreasing map. 
Alice and Bob aim to construct a protocol such that, given the measurement outcome $x$, the output state is close to $\Phi_{AB}^{\otimes \floor{nr_x^{(n)}}}$ for some $r_x^{(n)}$.
Namely, the rate $r_x^{(n)}$ of the target entangled state varies depending on the outcome $x$ and thus is a random variable.
Such a setting is known as variable-length coding in the context of data compression~\cite{Lempel-Ziv,ZivLempel1978,Hayashi2002quantum,Hayashi2002simpleconstructionquantumuniversal}, and here we employ this formulation for entanglement distillation~\cite{Matsumoto2007universal}.

In this setting, the error in the transformation can be characterized by the average error
\bal
 \varepsilon_{n}(\rho_{AB}) \coloneqq \sum_x p_x \left\|\frac{\Lambda_n^x(\rho_{AB}^{\otimes n})}{p_x}-\Phi_{AB}^{\otimes \floor{nr_x^{(n)}}}\right\|_1
\eal
where $p_x\coloneqq \Tr\qty[\Lambda_n^x(\rho_{AB}^{\otimes n})]$ is the probability of obtaining the outcome $x$. 
Our problem reduces to finding the rate function $r_x^{(n)}$ such that $\varepsilon_n(\rho_{AB})$ vanishes asymptotically for \emph{every} $\rho_{AB}$.
This motivates us to define an achievable rate of universal distillation as follows. 

\begin{defn}[Universally distillable entanglement] \label{defn:universal distillable entanglement main}
A function $R:\mD(\mH_{AB})\to\mbR_{\geq 0}$ is an achievable rate of universal entanglement distillation if, for an arbitrary $\delta>0$, there exists a series $\{\Lambda_n\}_{n}$ of an LOCC instrument $\Lambda_n=\{\Lambda_n^x\}_x$ and a series $\{r_x^{(n)}\}_{n,x}$ of real numbers such that the average error vanishes for every given state $\rho_{AB}$: 
\bal
 \varepsilon_{n}(\rho_{AB}) \xrightarrow[n\to\infty]{}0,\quad \forall\rho_{AB}\in\mD(\mH_{AB})
 \label{eq:error condition main}
\eal
and the probability of realizing the target rate $r_x^{(n)}$ smaller than $R(\rho_{AB})-\delta$ vanishes asymptotically:
\bal
 \Pr_x\qty[r_x^{(n)}<R(\rho_{AB})-\delta]\xrightarrow[n\to\infty]{}0,\quad \forall\rho_{AB}\in\mD(\mH_{AB}).
 \label{eq:overflow probability main}
\eal
\end{defn}

In the state-aware setting, Alice and Bob do not choose different target sizes, and thus do not need to consider instruments with various outcomes $x$.
Then, the standard definition of the state-aware distillable entanglement is recovered by removing the constraint $\forall\rho_{AB}\in\mD(\mH_{AB})$ in \eqref{eq:error condition main} and \eqref{eq:overflow probability main} and taking the supremum of the achievable rate.
Therefore, it is clear that any achievable rate $R$ of universal distillation satisfies 
\bal
 R(\rho_{AB})\leq E_d(\rho_{AB}),\quad \forall \rho_{AB}.
\eal
Our main result is to show that this equality is achievable at every state $\rho_{AB}$, i.e., the function $E_d$ is an achievable universal entanglement distillation rate. 

One might think that the equality should be achieved by the following simple protocol: use the first $o(n)$ copies of the input states to learn its description by state tomography and apply the protocol tailored to the estimated state to the rest of the input copies.  
Although this learn-and-apply protocol is powerful---and our construction indeed morally takes this path---this argument does not generally stand.
This is because of the potential blow-up of an error in the application stage.
The sample complexity of state tomography~\cite{ODonnell2016-1,Haah2017} implies that the use of $o(n)$ copies can learn the state with error $\epsilon=\omega(1/\sqrt{n})$ in trace distance, but then the protocol tailored to this imprecise estimation could cause errors up to $\epsilon (n-o(n))=\omega(\sqrt{n})$.
Therefore, one needs to employ a further structure of the setting to avoid this blow-up, which we accomplish here for entanglement distillation. 

We also remark that, even in the state-aware setting, computation of $E_d$ or constructing the optimal distillation strategy achieving $E_d$ is practically intractable except for a limited class of states.
Therefore, universal protocols can naturally be highly complex and may not be computationally tractable either. 
Our focus here instead is an information-theoretic existence statement, where the LOCC distillation strategy must be fixed independently of the unknown input state.

\subsection{Achieving hashing bound}

We begin by showing that coherent information---serving as a general lower bound for distillable entanglement as in \eqref{eq:hashing bound main}---can be achieved universally.  

\begin{pro}\label{pro:hashing bound universal main}
The coherent information $I_c(A|B)_{\rho_{AB}}$ is an achievable rate of universal entanglement distillation. 
\end{pro}
Even though the coherent information is the suboptimal entanglement distillation rate, it turns out that \cref{pro:hashing bound universal main} gives us a key observation toward constructing the universal entanglement distillation protocol achieving the distillable entanglement $E_d$ in the subsequent discussion.
In addition, the protocol achieving the coherent information is much simpler and tractable than the case of $E_d$, which has its own merit.
 
Fortunately, the core technical difficulty is already addressed in \cref{lem:universal state merging main}, which shows that the worst state-aware performance can be achieved in a state-agnostic manner. 
We lift this powerful result to the universal entanglement distillation protocol by combining it with local estimation of coherent information.
Alice and Bob begin by making a local informationally complete measurement on a sublinear number of input copies.
Alice then sends her outcomes to Bob, allowing him to get an estimate of the input state.
From this estimate, he also obtains an estimate of the coherent information $I_c(A|B)_{\rho_{AB}}$.
Bob sends the estimated value of coherent information back to Alice by sending $O(\log d_A + \log d_B + \log n)$ bits of classical information.
This allows them to consider the set $\mS$ of states whose coherent information is consistent with the estimated value within the accuracy ensured by state tomography.  
They run state merging with partial information in \cref{lem:universal state merging main} with this choice of $\mS$.
This allows them to distill entanglement at the rate of the minimum coherent information over the states in $\mS$.
Because of the continuity of the coherent information, by using a sufficiently large, but sublinear, number of input copies for the estimation of coherent information, the rate arbitrarily close to $I_c(A|B)_{\rho_{AB}}$ can be achieved.
A detailed proof can be found in Appendix~\ref{app:hashing universal}.

The state-aware protocol achieving the hashing bound only involves one-way LOCC, while our universal protocol involves $O(\log d_A + \log d_B + \log n)$ backward communication from Bob to Alice. 
Therefore, our protocol is not strictly one-way, but ``almost'' one-way in the sense that the \emph{rate} of backward communication still vanishes in the asymptotic limit.

\subsection{Achieving distillable entanglement}

Let us now present our main result.

\begin{thm}\label{thm:universal distillable entanglement main}
The distillable entanglement $E_d(\rho_{AB})$ is an achievable rate of universal entanglement distillation.
\end{thm}

While we defer the full proof to Appendix~\ref{app:distillable}, we here discuss the main idea of the construction. 
The starting point is to recall that the distillable entanglement can be formally written as~\cite{Devetak2005distillation} 
\bal
 E_d(\rho_{AB}) = \lim_{m\to\infty} \frac{E_d^{\leq}(\rho_{AB}^{\otimes m})}{m}
 \label{eq:distillable entanglement two-way formal main}
\eal
where
\bal
 E_d^{\leq}(\rho_{AB})\coloneqq \sup_{\Lambda\in{\rm LOCC}} I_c(A|B)_{\Lambda(\rho_{AB})}.
\eal
Although this form is practically not very useful because of the optimization over all LOCC operations and the regularized form, it is still useful for our purpose.
In particular, in light of \cref{pro:hashing bound universal main}, it appears plausible to first run the local state tomography to obtain an estimate $\widehat\rho_{AB}$,  apply the LOCC $\widehat\Lambda$ to the actual input $\rho_{AB}^{\otimes m}$ for a sufficiently large $m$, where $\widehat\Lambda$ is the LOCC strategy tailored to $\widehat\rho_{AB}^{\otimes m}$, 
and runs the protocol in \cref{pro:hashing bound universal main} by regarding $\widehat\Lambda(\rho_{AB}^{\otimes m})$ as one copy of the input state.

There are several issues in this idea we need to fix.  
The first problem is the continuity of $\sup_{\Lambda\in{\rm LOCC}}I_c(A|B)_{\Lambda(\rho_{AB})}$.
The underlying idea of using \cref{lem:universal state merging main} is to reduce the analysis of the achievable rate to the coherent information, allowing us to take a sufficiently small ball around the estimate to ensure that all states in the ball has similar values of coherent information.
However, for two close states $\rho_{AB}$ and $\rho_{AB}'$, one cannot exclude 
the possibility that their optimal LOCC operations $\Lambda$ and $\Lambda'$ are not close.
In particular, they could come with different output dimensions, 
making the size of $\Lambda(\rho_{AB})$ and $\Lambda'(\rho_{AB}')$ different.
Since the continuity of the coherent information is given with respect to the dimension, 
we cannot get a universal continuity bound for $\sup_{\Lambda\in{\rm LOCC}}I_c(A|B)_{\Lambda(\rho_{AB})}$
that would allow us to determine the required accuracy of the first state tomography.

To circumvent this, we consider another form of $E_d^\leq$~\cite{hayashi2016mathematical}
\bal
E_d^{\leq}(\rho_{AB})=\sup_{\Lambda_\mM\in{\mT_{\rm LOCC}}} I_c(A|BE)_{\Lambda_{\mM}(\rho_{AB})}
\label{eq:coherent info lower bound measurement}
\eal
where $\mM=\{M_j\}_j$ is a POVM on $A$ and $B$, $E$ is an auxiliary system to store the measurement outcome possessed by Bob, and $\mT_{\rm LOCC}$ is the set of LOCC channels of the form
\bal
 \Lambda_{\mM}(\rho_{AB}) = \sum_j p_j \frac{\sqrt{M_j}\rho_{AB}\sqrt{M_j}}{p_j} \otimes \dm{e_j}_E
 \label{eq:measurement channel}
\eal
where $p_j\coloneqq\Tr(M_j\rho_{AB})$.
Importantly, in this canonical representation,
each Kraus operator $\sqrt{M_j}$ is a square matrix and keeps the dimension of the original bipartite system $AB$ unchanged.
Since the measurement record on $E$ is held by Bob, $\Lambda_\mM$ does not change the dimension of system $A$.
Consequently, since the continuity of coherent information $I_c(A|BE)$ only involves the dimension of system $A$~\cite{Winter2016tight}, the coherent information can be made
robust by instead applying $\Lambda_{\widehat\mM}$ to $\rho_{AB}^{\otimes m}$ for sufficiently large $m$
where $\widehat M$ is the optimal LOCC measurement for the estimated state $\widehat\rho_{AB}$.

We still have another problem. 
How would Alice and Bob find a ``sufficiently large'' $m$ here? 
Although \eqref{eq:distillable entanglement two-way formal main} guarantees that,
for an arbitrary $\delta>0$, there is a sufficiently large $m$ such that $E_d^\leq(\rho_{AB}^{\otimes m})/m> E_d(\rho_{AB})-\delta$, 
this $m$ generally depends on $\rho_{AB}$, i.e., \eqref{eq:distillable entanglement two-way formal main} only 
guarantees the \emph{pointwise} convergence. 
This is problematic because here Alice and Bob do not know $\rho_{AB}$. 
In order for them to decide $m$ that is large enough to realize the target rate, 
\emph{uniform} convergence is needed. 
We resolve this by showing uniform convergence on arbitrary compact subsets of the set of full-rank distillable states.
This reduces the problem to finding such a compact set, the block size, and the final state-merging protocol.

This allows us to accomplish entanglement distillation for full-rank states, because the confidence set is eventually fully contained in such a compact subset. 
However, this may not be the case if $\rho_{AB}$ is not full rank, where the state resides on the boundary of the set of distillable states. 
We address this by perturbing the input state by adding noise to make the state full rank, while making the series of noise strengths vanish asymptotically so that the distillable entanglement of the original input state can be extracted. 
To confirm that this approach works, we also show that the distillable entanglement of perturbed states converges to that of the original input state. 

Our operational procedure taking into account the above considerations goes as 

\begin{enumerate}
    \item Use a predetermined sublinear number of copies for tomography
    \item \label{item:deciding parameter} Decide the noise strength $\xi$ for perturbation, the block size $m$ for state merging, and a finite-outcome LOCC measurement $\Lambda_{\widehat \mM}$ so that it will achieve the desired rate and fidelity
    \item  Perturb each remaining copy of $\rho_{AB}$ with the noise strength $\xi$ to get $\rho_{AB,\xi}$
    \item Apply the chosen LOCC measurement $\Lambda_{\widehat \mM}$ to each $\rho_{AB,\xi}^{\otimes m}$
    \item Apply state merging to the copies of $\Lambda_{\widehat\mM}\qty(\rho_{AB,\xi}^{\otimes m})$
\end{enumerate}

One might worry that, in Step~\ref{item:deciding parameter}, they need to be able to compute the fidelity of the state merging for chosen $\xi$, $m$, and $\Lambda_{\widehat\mM}$ without knowing $\rho_{AB}$. 
We show that they can actually do so because the fidelity guaranteed in \cref{lem:universal state merging main} can be bounded uniformly independent of the state information.

\section{Two-way quantum capacity of unknown channels}

We apply our result to the setting of quantum communication with LOCC assistance~\cite{Bennett-error-correction,pirandola_2017}. 
Let $\mN$ be a point-to-point channel from Alice to Bob. 
If they have access to LOCC for free, the natural quantity
to look at is the maximum size of the maximally entangled state that can be 
distributed by sequential use of $\mN$ and LOCC operation. 
Because of the LOCC assistance, distribution of the maximally entangled state is equivalent to noiseless
quantum communication because of quantum teleportation.

Formally, an $n$-use assisted protocol consists of an initial separable state, arbitrary local memory registers, and a sequence of LOCC maps interleaved with $n$ uses of $\mN$. 
Let $\rho_{AB}^{(n)}$ denote its final bipartite state. 
We say that the rate $C$ is achievable if such a sequence of adaptive protocols satisfies
\bal
  \norm{\rho_{AB}^{(n)}-\Phi_{AB}^{\otimes \floor{Cn}}}_1\xrightarrow[n\to\infty]{} 0.
\eal
Then, the two-way quantum capacity $C_{\leftrightarrow}(\mN)$ is defined as the supremum of achievable rates.

Two-way quantum capacity is generally hard to compute. 
One useful lower bound is given by the distillable entanglement of the Choi state
\bal
C_\leftrightarrow(\mN) \geq E_d(J_\mN)
\eal
where $J_\mN \coloneqq \id\otimes\mN(\Phi)$ is the 
Choi state of $\mN$. 
This is because Alice can always prepare a maximally entangled state locally and send one half of it to Bob
through $\mN$ to obtain $J_\mN$. They can then run entanglement distillation on the copy of the Choi state, distilling the rate of 
$E_d(J_\mN)$ of maximally entangled state. 

Notably, this bound becomes tight for a class of channels known as teleportation covariant channels~\cite{pirandola_2017,Kaur2017}. 
The channel $\mN$ is called teleportation covariant if 
for any generalized Pauli channel $\mP(\cdot)=P\cdot P^\dagger$ where $P$ is a multi-qubit Pauli operator for 
multi-qubit systems and a Heisenberg-Weyl operator for a qudit systems, it satisfies 
\bal
 \mN\circ\mP = \mU\circ\mN
\eal
for some unitary channel $\mU$.
Then, using the teleportation stretching technique~\cite{pirandola_2017},
one can see that 
\bal
 C_\leftrightarrow(\mN) = E_d(J_\mN).
\eal

Let us now suppose that $\mN$ is unknown.
This setting is naturally motivated by considering the situation where the description of the noise channel is unknown.
In particular, the capacity of channels whose description is uncertain has been studied in the context of universal channel coding~\cite{hayashi_2009-2,Hayashi_2022_universal_cq_superposition,Bjelakovic2008quantum,Berta_2017_compound_entanglement_assisted,Boche_entanglement-assisted_2017,Dasgupta2025,matsuura2025universalclassicalquantumchannelresolvability,watanabe2026needuniversalcorrelationdetector} and fault-tolerant quantum communication~\cite{ChristandlMullerHermes2024,ChristandlFawziGoswami2026}. 

In this case, we need to consider a series of LOCC protocols that do not depend on the description of $\mN$.
This motivates us to consider the setting of universal communication under LOCC assistance,
and define the universally achievable rate analogously to \cref{defn:universal distillable entanglement main}.
Namely, their LOCC operations are represented by an instrument, and 
they distill the maximally entangled state whose size is now described by a random variable. 
In the context of communication, this means that the number of qubits that can be reliably sent from Alice to Bob
is probabilistically determined.

Then, the following result is an immediate consequence of \cref{thm:universal distillable entanglement main}.

\begin{cor}
$E_d(J_\mN)$ is an achievable rate of universal quantum capacity
under LOCC assistance. 
Moreover, if it is known that $\mN$ is teleportation covariant,
the two-way capacity $C_{\leftrightarrow}(\mN)$ is universally achievable. 
\end{cor}

\section{Conclusions}

We showed that distillable entanglement can be achieved in entanglement distillation under LOCC, even when Alice and Bob do not know the input state at all. 
Our results do not assume any structure about the input state except that it is given in i.i.d. form, which is the standard setting of asymptotic entanglement distillation. 
Nevertheless, it is a natural and interesting question how far one could extend this observation to non-i.i.d. input sources. 
More broadly, clarifying the precise condition under which the state-aware rate can be achieved universally will deepen our understanding of universal resource manipulation. 
Another important direction is to study different figures of merit, such as the second-order asymptotic rate and the optimal rate of error decay, for which the lack of input-state knowledge is likely to be more sensitive than the linear rate of resource yield studied here. 

The technical backbone of our result is universal state merging, which is developed in the context of information processing on compound inputs. 
On the other hand, we did not employ another powerful technique using representation theory and Schur-Weyl duality in particular, which has proven to be useful in universal protocols in various other settings~\cite{Josza1998universal,Hayashi2002quantum,Matsumoto2007universal,matsuura2025universalclassicalquantumchannelresolvability,Watanabe_universal,watanabe2026needuniversalcorrelationdetector}, including the universal entanglement distillation with pure-state inputs~\cite{Matsumoto2007universal,Blume-Kohout2014}.
Devising a protocol that can fully employ the symmetry structure could provide a more efficient protocol. 
We leave a thorough investigation of these problems for future work.

\vspace{.5cm}

\emph{Note added.}---During the final stage of writing this manuscript, we became aware of an independent work by Tirone et al.~\cite{Tirone2026}, where they also showed the universal achievability of distillable entanglement by a different technique.
The main difference lies in the way of ensuring robust distillation after state tomography. 
We achieve this by employing state merging for compound sources, while they do so by running a hashing protocol on the output isotropic states with an unknown parameter.

\emph{Disclosure of AI use.}---The core idea and the initial proof employing state merging with partial information and uniform convergence was due to the authors. ChatGPT 5.6 Sol and 6 Astra were used to find technical mistakes and their revision, which were reviewed by the authors. 
The authors take full responsibility for the contents of the manuscript.

\emph{Acknowledgments.}---R.T. thanks Ludovico Lami for helpful discussions.
This work was supported by JSPS KAKENHI Grant Nos.\ 24K16975, 25K00924, 25KJ0043, 26H02015, 26KJ0965, the Japan Science and Technology Agency (JST) CREST Grant No.\ JPMJCR23I3, NEXUS Grant Number JPMJNX26C2, PRESTO Grant No.\ JPMJPR24FA, RIKEN iTHEMS, RIKEN Pioneering Project `Mathematical foundation of quantum information' (PI Yasuyuki Kawahigashi), and the World-Leading Innovative Graduate Study Program for Advanced Basic Science Course (WINGS-ABC) at the University of Tokyo.

\bibliographystyle{apsrmp4-2}
\bibliography{myref}

\let\addcontentsline\oldaddcontentsline

\clearpage
\newgeometry{hmargin=1.2in,vmargin=0.8in}

\appendix
\widetext

\setcounter{tocdepth}{1}
\tableofcontents

\section{Preliminaries}

\subsection{Entanglement distillation}

Here, we review the standard setting of entanglement distillation from known i.i.d. states. 
Let $\mD(\mH)$ denote the set of all states acting on the Hilbert space $\mH$.  
We consider the setting where two parties, Alice and Bob, possess a bipartite state $\rho_{AB}\in\mD(\mH_{AB})$, where $\mH_{AB}=\mH_A\otimes \mH_B$ with $\mH_A$ and $\mH_B$ being finite-dimensional Hilbert spaces with dimensions $d_A$ and $d_B$. 
The distillable entanglement $E_d(\rho_{AB})$ of the state $\rho_{AB}$ is defined by 
\bal
 E_d(\rho_{AB})\coloneqq \sup\lset r\geq 0 \sbar \exists \{\Lambda_n\}_n\in{\rm LOCC},\ \left\|\Lambda_n(\rho_{AB}^{\otimes n})-\Phi_{AB}^{\otimes \lfloor rn \rfloor}\right\|_1\xrightarrow[n\to \infty]{} 0\rset.
 \label{eq:distillable entanglement def}
\eal

Throughout this manuscript, LOCC denotes the set of local operations and classical communication with finite rounds and finite measurement outcomes. 
As can be seen in the definition, the distillable entanglement is given as a result of the optimization over all possible LOCC protocols \emph{for the given $\rho_{AB}$}. 
This particularly means that the optimal series $\{\Lambda_n\}_n$ of LOCC protocols can generally depend on the description of $\rho_{AB}$.
In other words, this is the ultimate performance under the assumption that the experimenter has complete knowledge about the description of $\rho_{AB}$.

The distillable entanglement is a fundamental quantifier of entanglement that directly reflects the operational significance of the entanglement contained in the given state. 
On the other hand, a closed form for the distillable entanglement is not known, and computing $E_d$ is generally a hard problem.
This also motivates us to obtain a more tractable achievable lower bound that applies to an arbitrary state $\rho_{AB}$. 
The standard lower bound is the \emph{hashing bound}~\cite{Devetak2005distillation,Horodecki2007quantum}, giving 
\bal
 E_d(\rho_{AB})\geq I_c(A|B)_{\rho_{AB}}
 \label{eq:hashing bound}
\eal
where $I_c(A|B)_{\rho_{AB}} \coloneqq H(\rho_B)-H(\rho_{AB})$ is the coherent information, and $H(\rho)\coloneqq-\Tr(\rho\log\rho)$ is the von Neumann entropy of the state $\rho$.
In fact, the hashing bound can be achieved by one-way LOCC---LOCC that involves only classical communication from Alice to Bob. 
One of the specific protocols achieving this bound plays a central role in our result, and it is based on the task known as state merging, which we review in the following.

\subsection{State merging}

State merging is a quantum information processing task where Alice aims to transfer her part of the known given state to Bob while maintaining the correlation with the environment by LOCC and entanglement.
More specifically, they are given the initial state $\rho_{AB}^{\otimes n}$ with purification $\psi_{ABE}^{\otimes n}$.
They also begin with the initial entanglement $\Phi_{d_{\rm ini}}$ and end up with the final entanglement $\Phi_{d_{\rm out}}$.
Then, if there exists a sequence $\{\Lambda_n\}_n$ of LOCC operations such that
\bal
 \norm{\Lambda_n(\psi_{ABE}^{\otimes n}\otimes\Phi_{d_{\rm ini}^{(n)}})-\psi_{\tilde B BE}^{\otimes n}\otimes \Phi_{d_{\rm out}^{(n)}}}_1\xrightarrow[n\to\infty]{} 0,
\eal
we say that $\limsup_{n\to \infty}\frac{\log d_{\rm ini}^{(n)}-\log d_{\rm fin}^{(n)}}{n}$ is an achievable entanglement cost for state merging (Fig.~\ref{fig:state_merging}).

\begin{figure}[htbp]
    \centering
    \includegraphics[width=0.45\linewidth]{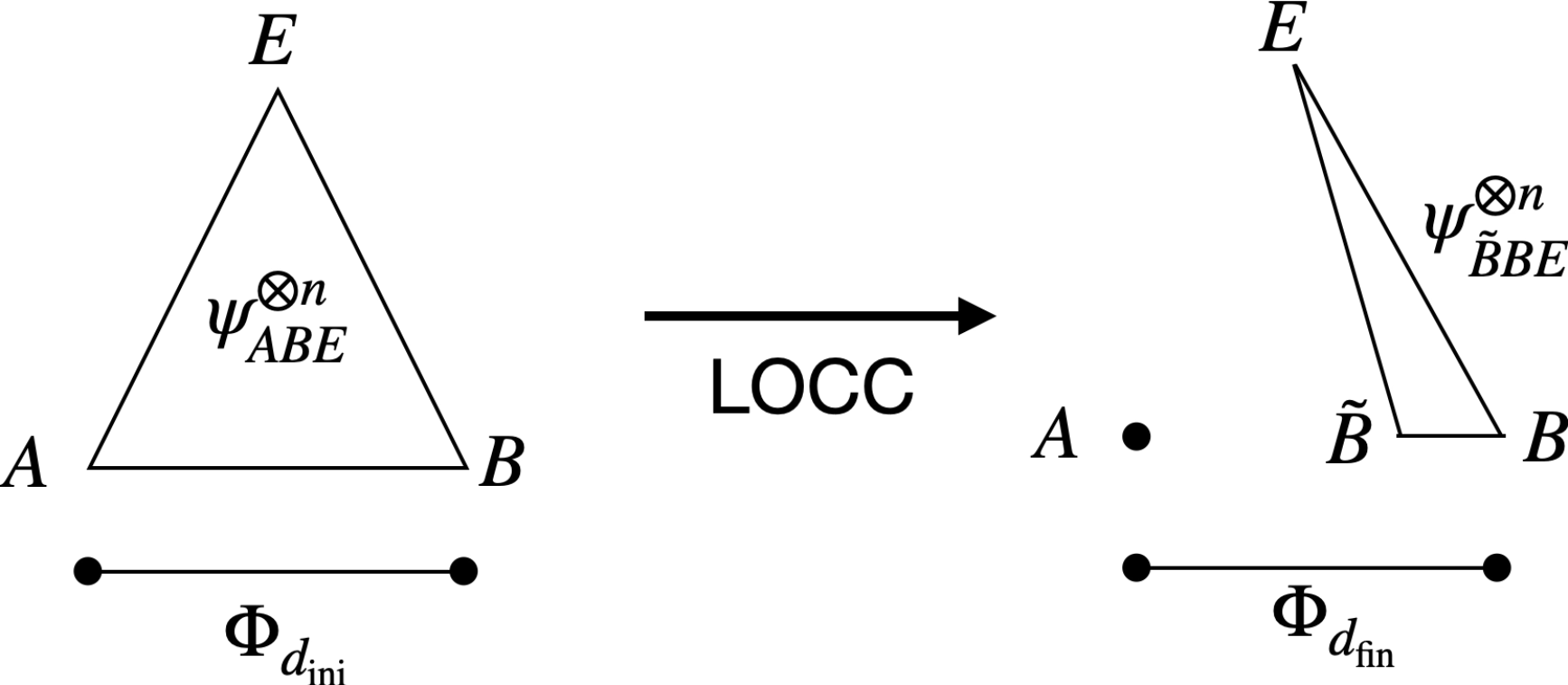}
    \caption{Setting of asymptotic state merging}
    \label{fig:state_merging}
\end{figure}

It is known that the optimal entanglement cost is given by the conditional entropy $H(A|B)_{\rho_{AB}} = H(\rho_{AB})-H(\rho_B)$, which can be achieved by a measurement on Alice's side, followed by Bob's conditional operation~\cite{Horodecki2007quantum}.
Interestingly, this also encompasses the case where the conditional entropy is negative.
In this case, we can rather extract entanglement with the rate of coherent information $I_c(A|B)_{\rho_{AB}}=-H(A|B)_{\rho_{AB}}$.
Moreover, a close look at the construction in Ref.~\cite{Horodecki2007quantum} reveals that when $I_c(A|B)_{\rho_{AB}}>0$, one can start with no initial entanglement, i.e., $d_{\rm ini}=1$.
This shows that, as a ``byproduct'' of the state merging task, one can realize entanglement distillation with input state $\rho_{AB}$ without initial entanglement with the rate $I_c(A|B)_{\rho_{AB}}$ of distilled entanglement.

\subsubsection{State merging with partial information}

Crucially, the result in Ref.~\cite{Horodecki2007quantum} assumes that Alice and Bob know the description of the state beforehand.
The relevant setting of our interest is the case when the description of $\rho_{AB}$ is not known.
An extension to the setting where the full description of $\rho_{AB}$ is not available was considered in Ref.~\cite{Bjelakovic2013universal}. 
They considered the setting where the state $\rho_{AB}$ is taken from the known subset $\mS$ of quantum states.
Namely, the experimenter has access to \emph{partial} information about the state $\rho_{AB}$ in the sense that they know the base set to which the given state belongs.
In this setting, they investigated the required entanglement cost so that there exists a fixed protocol that accomplishes the merging task for every state in $\mS$.
Their result shows that the entanglement cost is characterized by the maximum conditional entropy of the state in the set.

\begin{lem}[State merging with partial information~\cite{Bjelakovic2013universal,Boche2014resource}]\label{lem:universal state merging}
    Let $\mS$ be an arbitrary set of bipartite states on the given system $AB$ with purifying system $E$, and let $d_{AB}$ denote the dimension of the Hilbert space underlying the system $AB$. 
    Then, for every $\delta$ satisfying $0<\delta<\inf_{\rho_{AB}\in\mS}I_c(A|B)_{\rho_{AB}}$, and every integer $n>n_0(\delta,d_{AB})$ where $n_0$ is non-increasing with $\delta$ and non-decreasing with $d_{AB}$, there exists a series $\{\Lambda_n\}_n$ of one-way LOCC protocol from Alice to Bob, a sequence $\{d_{\rm out}^{(n)}\}_n$ of integers, and a positive function $c(\delta,d_{AB})$, which is non-decreasing with $\delta$ and non-increasing with $d_{AB}$, such that 
    \bal
     \norm{\Lambda_n(\psi_{ABE}^{\otimes n})-\psi_{\tilde B BE}^{\otimes n}\otimes \Phi_{d_{\rm out}}}_1\leq 2^{-nc(\delta,d_{AB})}\quad \forall \rho_{AB}\in\mS
     \label{eq:state merging partial info fidelity convergence}
    \eal
    with
    \bal
     \frac{\log d_{\rm out}^{(n)}}{n} \geq \inf_{\rho_{AB}\in\mS} I_c(A|B)_{\rho_{AB}}-\delta.
    \eal

\end{lem}

\cref{lem:universal state merging} particularly means that entanglement distillation from an unknown state from the set $\mS$ can be accomplished with the rate arbitrarily close to $\min_{\rho_{AB}\in\mS}I_c(A|B)_{\rho_{AB}}$.
We also remark that the corresponding result in \cite{Boche2014resource} is stated in terms of the fidelity, from which our trace-distance bound follows by Fuchs-van de Graaf inequality~\cite{FuchsVanDeGraaf1999}.
Although they express the exponent $c$ as a function of the set $\mS$, going through their proof reveals that it only depends on $\delta$ and dimension $d_{AB}$.

\subsection{Continuity of coherent information}
\label{sec:continuity coherent information}

We remark here the well-known continuity bound of coherent information. 
Namely, for an arbitrary state $\sigma_{AB}$ and $\tilde\sigma_{AB}$ such that $\frac{1}{2}\norm{\sigma_{AB}-\tilde\sigma_{AB}}_1\leq \epsilon$, it satisfies~\cite{Winter2016tight} 
\bal
 \abs{I_c(A|B)_{\sigma_{AB}}-I_c(A|B)_{\tilde\sigma_{AB}}} \leq 2\epsilon\log d_A + (1+\epsilon)h\qty(\frac{\epsilon}{1+\epsilon}),
 \label{eq:continuity coherent information}
\eal
where $h(p)\coloneqq -p\log p -(1-p)\log(1-p)$ is the binary entropy.
Notably, the right-hand side depends on the dimension of the conditioned system only.


\section{Universal entanglement distillation}

\subsection{Setting}

In the following, we formalize the setting where the experimenter does not know the state $\rho_{AB}$. 
The idea is that, just like variable-length source coding~\cite{Hayashi2002quantum,Bennett2006universal}, the protocol outputs a state close to the maximally entangled state, whose size is a random variable.
Then, if the probability of having size larger than a certain number approaches unity while the error vanishes in the asymptotic limit, we can understand this number as an achievable rate.
If this rate can be chosen as a function of the input state, then it can be regarded as the universal entanglement distillation rate.
We formally define this notion as follows.

\begin{defn}[Universally distillable entanglement] \label{defn:universal distillable entanglement}
A function $R:\mD(\mH_{AB})\to\mbR_{\geq 0}$ is an achievable rate of universal entanglement distillation if, for an arbitrary $\delta>0$, there exists a series $\{\Lambda_n\}_{n}$ of an LOCC instrument $\Lambda_n=\{\Lambda_n^x\}_x$ and a series $\{r_x^{(n)}\}_{n,x}$ of real numbers such that the average error vanishes for every given state $\rho_{AB}$: 
\bal
 \varepsilon_{n}(\rho_{AB}) \coloneqq \sum_x p_x \left\|\frac{\Lambda_n^x(\rho_{AB}^{\otimes n})}{p_x}-\Phi_{AB}^{\otimes \floor{nr_x^{(n)}}}\right\|_1\xrightarrow[n\to\infty]{}0,\quad \forall\rho_{AB}\in\mD(\mH_{AB})
 \label{eq:error condition}
\eal
where $p_x \coloneqq \Tr\qty(\Lambda_n^x(\rho_{AB}^{\otimes n}))$, and the probability of realizing the target rate $r_x^{(n)}$ smaller than $R(\rho_{AB})-\delta$ vanishes asymptotically:
\bal
 \Pr_x\qty[r_x^{(n)}<R(\rho_{AB})-\delta]\xrightarrow[n\to\infty]{}0,\quad \forall\rho_{AB}\in\mD(\mH_{AB}).
 \label{eq:overflow probability}
\eal
\end{defn}

A significant question is whether universally distillable entanglement can coincide with the \emph{state-aware} performance, where experimenters can tailor the distillation protocol.
This problem was studied in Ref.~\cite{Matsumoto2007universal} for the case of pure unknown states, where it was shown that the optimal rate of von Neumann entropy of entanglement $S_E(\psi_{AB})=H(\Tr_B\psi_{AB})$ is achievable. 
In this work, we extend this finding to mixed states, i.e., to the setting where no information about the state is provided beforehand.

\subsection{Universal protocol achieving the hashing bound (Proof of \cref{pro:hashing bound universal main})}\label{app:hashing universal}

Here, we first show that the hashing bound \eqref{eq:hashing bound} is universally achievable.
Unlike the original state-aware entanglement distillation protocol~\cite{Devetak2005distillation,Horodecki2007quantum}---which can be accomplished by \emph{one-way} LOCC, we allow ourselves to use two-way LOCC. 
Nevertheless, only one round of classical communication suffices, and the rate of required classical communication from Bob to Alice vanishes in the limit of $n\to\infty$. 

\begin{thm}\label{thm:universal entanglement distillation hashing}
  The hashing bound $I_c(A|B)_{\rho_{AB}}$ is an achievable universally distillable entanglement.
  Moreover, it can be realized by almost one-way communication, where $O(\log d_A + \log d_B + \log n)$ bits of backward classical communication are sufficient for input $\rho_{AB}^{\otimes n}$.
\end{thm}

\begin{proof}

We prove Theorem~\ref{thm:universal entanglement distillation hashing} by lifting the protocol in Lemma~\ref{lem:universal state merging} to the universal entanglement distillator that distills the rate of coherent information of the given unknown state.  
For each positive integer $n$, our protocol for $n$ inputs proceeds as follows.
In the following, we set 
\bal
\epsilon_n=\eta_n=n^{-1/3}.
\eal

\stepparagraph{Estimate the coherent information.}\label{para:hashing tomography}

Alice and Bob first estimate the coherent information of the input state by LOCC.
This is the only step where backward communication from Bob to Alice is required.

To obtain the estimate, they run state tomography by LOCC using  $\alpha_n=O\left(\frac{d_{AB}^4 \log(1/\eta_n)}{\epsilon_n^2}\right)=o(n)$ input copies.
It is known that this number of copies suffices to ensure the estimation with trace-distance accuracy $\epsilon_n$ with probability at least $1-\eta_n$ by local two-outcome measurements on each subsystem, followed by classical communication~\cite{Lowe2025lower}.
Here, Alice first sends her measurement outcomes to Bob, from which Bob obtains the classical estimation $\widehat\rho_{AB,x}$ of the input state, where $x$ denotes the measurement outcomes that Alice and Bob obtained.  

Using the estimation of the input state, Bob computes its coherent information $\widehat I_x$. 
Bob then sends the value of $\widehat I_x$ to Alice by backward communication. 
For the estimate $\widehat\rho_{AB,x}$ such that $\frac{1}{2}\|\rho_{AB}-\widehat\rho_{AB,x}\|_1\leq \epsilon_n$, the continuity of coherent information in \eqref{eq:continuity coherent information} guarantees that 
\bal
 \abs{I_c(A|B)_{\rho_{AB}}-\widehat I_x}\leq 2\epsilon_n\log d_A + (1+\epsilon_n)h\qty(\frac{\epsilon_n}{1+\epsilon_n})=:\epsilon'_n.
 \label{eq:estimated coherent information continuity}
\eal
Also, since $I_c(A|B)_{\rho_{AB}}<d_{AB}$, it suffices for Bob to transmit $O(\log d_{AB}+\log(1/\epsilon'_n))=O(\log d_A + \log d_B + \log n)$ classical bits to Alice to share the $\epsilon'_n$-approximation of $I_c(A|B)_{\rho_{AB}}$.

\stepparagraph{State merging to distill entanglement.}\label{para:hashing state merging}

Alice and Bob fix an arbitrary positive number $\delta>0$ corresponding to the target deviation in \eqref{eq:overflow probability}, and choose the state-merging parameter to be $\delta/2$.
If $\widehat I_x-\epsilon'_n-\frac{\delta}{2}\leq 0$, they output a product state, yielding zero distilled entanglement. 

If $\widehat I_x-\epsilon'_n-\frac{\delta}{2}>0$, they consider the set of states whose coherent information is $\epsilon'_n$-close to the estimated value, defined by
\bal
 \mS^{\epsilon'_n}_x \coloneqq \lset \sigma_{AB}\sbar \left|I_c(A|B)_{\sigma_{AB}}-\widehat I_x\right|\leq\epsilon'_n\rset.
\eal
They then run the state merging protocol in Lemma~\ref{lem:universal state merging} with the choice of $\mS = \mS^{\epsilon'_n}$ to the remaining $n-\alpha_n$ copies of input states.

\vspace{.5cm}

We now check that the above protocol satisfies \eqref{eq:error condition} and \eqref{eq:overflow probability} in \cref{defn:universal distillable entanglement} with $I_c(A|B)_{\rho_{AB}}$ being a universally achievable rate.

Here, we take $x$ to be the tomography data that estimates $\widehat\rho_{AB,x}$, and $\Lambda_n^x$ is a subchannel that obtains the tomography data $x$ from the first $\alpha_n$ copies of $\rho_{AB}$ followed by the state merging protocol applied to the remaining $n-\alpha_n$ copies of $\rho_{AB}$.
Let $\delta>0$ be the positive constant fixed in \ref{para:hashing state merging} and take
\bal
r_x^{(n)}=\frac{n-\alpha_n}{n}\max\qty{\widehat I_x - \epsilon_n' -\frac{\delta}{2},0}.
\label{eq:target rate hashing}
\eal 
Let $\tilde X$ be the set of $x$'s that correspond to $\widehat\rho_{AB,x}$ such that $\frac{1}{2}\norm{\widehat\rho_{AB,x}-\rho_{AB}}_1\leq \epsilon_n$. 
Then, the average error in \eqref{eq:error condition} can be estimated as
\bal
\varepsilon_n(\rho_{AB})&= \sum_{x\in\tilde X} p_x\norm{\frac{\Lambda_n^x(\rho_{AB}^{\otimes n})}{p_x}-\Phi_{AB}^{\otimes \floor{nr_x^{(n)}}}}_1 + \sum_{x\not\in\tilde X} p_x\norm{\frac{\Lambda_n^x(\rho_{AB}^{\otimes n})}{p_x}-\Phi_{AB}^{\otimes \floor{nr_x^{(n)}}}}_1\\
&\leq (1-\eta_n)\max_{x\in\tilde X}\norm{\frac{\Lambda_n^x(\rho_{AB}^{\otimes n})}{p_x}-\Phi_{AB}^{\otimes \floor{nr_x^{(n)}}}}_1 + 2\eta_n.
\label{eq:error hashing}
\eal 
We then know from \cref{lem:universal state merging} that the first term in \eqref{eq:error hashing} vanishes in the $n\to\infty$ limit. 
Since $\eta_n\xrightarrow[n\to\infty]{}0$, the second term also vanishes, showing \eqref{eq:error condition}.

We finally see that the choice of \eqref{eq:target rate hashing} allows us to take $R=I_c(A|B)_{\rho_{AB}}$ in \eqref{eq:overflow probability}.
Recalling that $\widehat\rho_{AB,x}$ is an estimate of $\rho_{AB}$ given by state tomography in \ref{para:hashing tomography}, it is ensured that $\rho_{AB}\in\mS_x^{\epsilon'_n}$ with probability at least $1-\eta_n$, which approaches 1 as $n\to\infty$.
In this case, it holds that 
\bal
r_x^{(n)}- I_c(A|B)_{\rho_{AB}} &\geq r_x^{(n)}- \widehat I_x - \epsilon_n'\\
&\geq-\frac{\alpha_n}{n}\widehat I_x -\qty(2+\frac{\alpha_n}{n})\epsilon_n'-\qty(1-\frac{\alpha_n}{n})\frac{\delta}{2}
\eal
where the first line is because of \eqref{eq:estimated coherent information continuity}, and the second line follows from \eqref{eq:target rate hashing}.
Recalling $\alpha_n=o(n)$ and $\epsilon_n'=o(n)$, the right-hand side converges to $-\frac{\delta}{2}$ in $n\to\infty$ limit. 
Taking a sufficiently small $\delta$ confirms that \eqref{eq:overflow probability} holds.
\end{proof}

\subsection{Universal protocol achieving distillable entanglement (Proof of \cref{thm:universal distillable entanglement main})}\label{app:distillable}

Here, we prove \cref{thm:universal distillable entanglement main} in the main text, stating that distillable entanglement is universally achievable.
\cref{thm:universal distillable entanglement main} is a direct consequence of \cref{pro:analysis of protocol} presented in Sec.~\ref{sec:analysis}.
Toward this, we first provide the main idea underlying the proof and several technical results on the uniform convergence, followed by an explicit description of our protocol and its analysis.

\subsubsection{Overall idea}\label{sec:idea}

Before digging into the detailed proof, let us first provide an overall idea behind our protocol. 
We first recall that the distillable entanglement $E_d(\rho_{AB})$ under two-way LOCC can formally be characterized as~\cite{Devetak2005distillation}
\bal
 E_d(\rho_{AB}) = \lim_{m\to\infty} \frac{E_d^{\leq}(\rho_{AB}^{\otimes m})}{m}
 \label{eq:distillable entanglement two-way formal}
\eal
where
\bal
 E_d^{\leq}(\rho_{AB})\coloneqq \sup_{\Lambda\in{\rm LOCC}} I_c(A|B)_{\Lambda(\rho_{AB})}.
 \label{eq:one-shot distillable entanglement coherent information supremum LOCC}
\eal
We note that this supremum may not be attainable.
Nevertheless, for any $\delta>0$,  there exists an LOCC that achieves $E_d^\leq(\rho_{AB})-\delta$.

This form appears useful for our purpose, suggesting the potential of applying the idea used in \cref{thm:universal entanglement distillation hashing} to the block $\Lambda_m(\rho_{AB}^{\otimes m})$, where $\Lambda_m$ is an optimal LOCC for $m$ copy state $\rho_{AB}^{\otimes m}$.
Namely, we first run local state tomography to learn $\rho_{AB}$, prepare many copies of $\Lambda_m(\rho_{AB}^{\otimes m})$ for some sufficiently large $m$, and apply state merging with partial information in \cref{lem:universal state merging} to the set of states close to $\Lambda_m(\rho_{AB}^{\otimes m})$ whose size is set by the accuracy of tomography.
We would then use the continuity of coherent information to realize the rate $E_d(\rho_{AB})-\delta$ for an arbitrary $\delta>0$.

However, there is a subtlety in the last step, when using the continuity of coherent information.
The continuity bound of coherent information involves the dimension of the output of $\Lambda_m$, but we do not have a universal upper bound on the output dimension of the optimal LOCC. 
Therefore, there is always a chance that, after state tomography using a certain number of copies and obtaining the state estimate and the corresponding LOCC, they get to know that the accuracy of tomography is not enough to realize the distillation rate within the target deviation. 
They could continue further state tomography, but then it gives a different state estimate and LOCC, which could come with an even larger output dimension, making the convergence of this process unclear.

We circumvent this issue by employing the canonical measurement representation of $E_d^\leq$ derived in Ref.~\cite[Eqs.~(8.125) and (8.126)]{hayashi2016mathematical}
\bal
E_d^{\leq}(\rho_{AB})=\sup_{\Lambda_\mM\in \mT_{\rm LOCC}} I_c(A|BE)_{\Lambda_{\mM}(\rho_{AB})}
\label{eq:another form LOCC measurement}
\eal
where $\mM=\{M_j\}_j$ is a POVM on the bipartite system $AB$ associated with this canonical measurement representation, and $\mT_{\rm LOCC}$ is the set of LOCC channels of the form
\bal
 \Lambda_{\mM}(\rho_{AB}) = \sum_j p_j \frac{\sqrt{M_j}\rho_{AB}\sqrt{M_j}}{p_j} \otimes \dm{e_j}_E,
\eal
which is the channel that measures $A$ and $B$ with the POVM $\mM$ and stores the measurement outcomes in another system $E$ held by Bob.
The point of this rewriting is to keep the dimension of Alice's subsystem invariant---as POVM elements $\{M_j\}$ are square matrices, the input and output dimensions on Alice's side coincide. 
We do not have control over the dimension of Bob's output $BE$, but it is not a problem here because the continuity bound of the coherent information only involves the dimension of the \emph{conditioned} system~\cite{Winter2016tight}. 
We remark on this in Sec.~\ref{sec:continuity coherent information}.

This suggests the following preliminary protocol. 
We first run the local state tomography to get a good estimate $\widehat\rho_{AB}$ of the input state $\rho_{AB}$ using a sublinear number $\alpha$ of copies. 
Then, we divide the remaining $n-\alpha$ copies into blocks of $m$ copies and apply the optimal LOCC $\Lambda_{\widehat\mM}$ tailored to the estimate $\widehat\rho_{AB}$ to each block.
The state in each block has the form $\Lambda_{\widehat\mM}(\sigma_{AB}^{\otimes m})$ for some $\sigma_{AB}\in\mB(\widehat\rho_{AB})$ where $\mB(\widehat\rho_{AB})$ is a neighborhood around the estimate $\widehat\rho_{AB}$, and the total state has the i.i.d. form $\qty[\Lambda_{\widehat\mM}(\sigma_{AB}^{\otimes m})]^{\otimes (n-\alpha)/m}$.
Therefore, by applying state merging in \cref{lem:universal state merging} to the set 
\bal
 \lset \Lambda_{\widehat\mM}(\sigma_{AB}^{\otimes m})\sbar \sigma_{AB}\in\mB(\widehat\rho_{AB})\rset,
\eal
this achieves the rate 
\bal
\min_{\sigma_{AB}\in\mB(\widehat\rho_{AB})}\frac{I_c(A|BE)_{\Lambda_{\widehat \mM}(\sigma_{AB}^{\otimes m})}}{m}.
\eal
Therefore, noting 
 \eqref{eq:distillable entanglement two-way formal}, this strategy appears to achieve the rate $E_d(\rho_{AB})-\delta$ for an arbitrary $\delta>0$ by taking sufficiently large $m$ and $\alpha$ in the limit $n\to\infty$.  

However, we still have another subtlety here. 
Recall that we do not know the state $\rho_{AB}$, and the ``large enough'' $m$ can generally depend on $\rho_{AB}$. 
In order for us to choose a sufficiently large $m$ such that $E_d^\leq(\rho_{AB}^{\otimes m})/m\geq E_d(\rho_{AB})-\delta$ for \emph{every} state $\rho_{AB}$, the series $\qty{E_d^\leq(\rho_{AB}^{\otimes m})/m}_m$ of functions needs to converge to $E_d(\rho_{AB})$ \emph{uniformly}, meaning that for an arbitrary $\delta>0$, there is a sufficiently large $m$ such that $E_d^\leq (\rho_{AB}^{\otimes m})/m\geq E_d(\rho_{AB})-\delta$ for every state $\rho_{AB}$.

We resolve this problem by showing that the series $\qty{E_d^\leq(\rho_{AB}^{\otimes m})/m}_m$ converges uniformly to $E_d(\rho_{AB})$ on every compact subset of the set of full-rank distillable states; see Sec.~\ref{sec:uniform convergence}.  
Once a confidence set obtained by the initial tomography lies in such a compact region, the theorem gives a block size $m$ that works uniformly over the entire confidence set.

This allows us to accomplish entanglement distillation for full-rank states, because the confidence set is eventually fully contained in such a compact subset. 
However, this may not be the case if $\rho_{AB}$ is not full rank, where the state resides on the boundary of the set of distillable states. 
We address this by perturbing the input state by adding depolarizing noise to make the state full rank, while making the series of noise strengths vanish asymptotically so that the distillable entanglement of the original input state can be extracted. 
To confirm that this approach works, we show in Sec.~\ref{sec:robustness} that the distillable entanglement of perturbed states converges to that of the original input state. 

Our operational procedure taking into account the above considerations goes as 

\begin{enumerate}
    \item Use a predetermined sublinear number of copies for tomography
    \item \label{item:deciding parameter} Decide the noise strength $\xi$ for perturbation, the block size $m$ for state merging, and a finite-outcome LOCC measurement $\Lambda_{\widehat \mM}$ so that it will achieve the desired rate and fidelity
    \item  Perturb each remaining copy of $\rho_{AB}$ with the noise strength $\xi$ to get $\rho_{AB,\xi}$
    \item Apply the chosen LOCC measurement $\Lambda_{\widehat \mM}$ to each $\rho_{AB,\xi}^{\otimes m}$
    \item Apply state merging to the copies of $\Lambda_{\widehat\mM}\qty(\rho_{AB,\xi}^{\otimes m})$
\end{enumerate}

One might worry that, in Step~\ref{item:deciding parameter}, they need to be able to compute the fidelity of the state merging for chosen $\xi$, $m$, and $\Lambda_{\widehat\mM}$ without knowing $\rho_{AB}$. 
We see that they can actually do so because the fidelity guaranteed in \cref{lem:universal state merging} can be bounded uniformly independent of the state information. 

In Sec.~\ref{sec:protocol}, we provide each step in the procedure in detail, and in Sec.~\ref{sec:analysis}, we show that one can indeed find the required parameters in Step~\ref{item:deciding parameter}, and that the whole procedure serves as the universal entanglement distillation protocol satisfying the requirements in \cref{defn:universal distillable entanglement}.

\subsubsection{Uniform convergence of $E_d^\leq$}\label{sec:uniform convergence}

Here, we show a uniform convergence property of $E_d^\leq$ (\cref{cor:uniform continuity of distillable entanglement}), which plays a key role in the universal distillation protocol. 
To this end, we first remark the continuity of $E_d$ on the set of full-rank distillable states.

\begin{lem}\label{lem:pointwise continuity of distillable entanglement}
Let $\mbD_{>0}\coloneqq \lset \rho_{AB}\sbar E_d(\rho_{AB})> 0,\ \lambda_{\min}(\rho_{AB})> 0\rset$, where $\lambda_{\min}(\rho_{AB})$ is the minimum eigenvalue of $\rho_{AB}$, be the set of full-rank distillable states.
Then, $E_d(\rho_{AB})=\lim_{m\to\infty}\frac{E_d^\leq(\rho_{AB}^{\otimes m})}{m}$ is continuous on $\mbD_{>0}$. 
\end{lem}
\begin{proof}
In Ref.~\cite{vidal2002continuityasymptoticmeasuresentanglement}, it is shown that the distillable entanglement is continuous in any open subset of the set of distillable states.
On the other hand, since the set of undistillable states is closed~\cite{horodecki2001boundentanglementcontinuousvariables} and the set of full-rank states is in the interior of the set of quantum states, the set $\mbD_{>0}$ is an open subset, which concludes the proof.
\end{proof}

The following result shows that pointwise convergence can be lifted to uniform convergence under several conditions.

\begin{lem}\label{lem:lifting pointwise to uniform continuity}
Let $\mD(\mH)$ denote the set of states on $\mH$.
Let $\mS\subset \mD(\mH)$ be a compact set and $f_m$ be a non-negative function over $\mD(\mH^{\otimes m})$.
Suppose that $f_m$ satisfies the following conditions.
\begin{enumerate}
    \item $f_m$ is continuous for every positive integer $m$ 
    \item Superadditivity: $f_{m+n}(\rho^{\otimes (m+n)})\geq f_m(\rho^{\otimes m}) + f_n(\rho^{\otimes n})$
    \item The limit $f^\infty(\rho)=\lim_{m\to\infty}\frac{f_m(\rho^{\otimes m})}{m}$ (which is guaranteed to exist because of Fekete's lemma) is continuous on $\mS$
\end{enumerate}
Then, the series $\qty{\frac{f_m(\rho^{\otimes m})}{m}}_m$ converges to $f^\infty(\rho)$ uniformly.
\end{lem}
\begin{proof}

Since $f_m$ is superadditive, Fekete's lemma ensures that 
\bal
 f^\infty(\rho) = \sup_m \frac{f_m\qty(\rho^{\otimes m})}{m}
\eal
for every $\rho\in\mD(\mH)$.
This implies that for every integer $m$ and every $\rho$, it holds that
\bal
 f^\infty(\rho)\geq \frac{f_m(\rho^{\otimes m})}{m}.
\eal
Therefore, it suffices to show that, for every $\delta>0$, there is a sufficiently large integer $m$---which is independent of the input $\rho$---such that 
\bal
 f^\infty(\rho)-\frac{f_m\qty(\rho^{\otimes m})}{m}\leq \delta,\quad \forall \rho.
\eal

Let $\delta>0$ be arbitrary.
The pointwise convergence of $\frac{f_m(\rho^{\otimes m})}{m}$ guarantees that, for every $\rho\in\mS$, there exists an integer $m_\rho$ such that 
\bal
 f^\infty(\rho)-\frac{f_{m_\rho}(\rho^{\otimes m_\rho})}{m_\rho}\leq\delta/3.
\eal
Since $f_{m_\rho}$ is continuous on $\mD(\mH^{\otimes m_{\rho}})$ and $f^\infty$ is continuous on $\mS$, there exists a neighborhood $\mB(\rho)\subset \mS$ of $\rho$ such that, for every $\sigma\in\mB(\rho)$, it satisfies 
\bal
\abs{\frac{f_{m_\rho}(\sigma^{\otimes m_\rho})}{m_\rho}-\frac{f_{m_\rho}(\rho^{\otimes m_\rho})}{m_\rho}}\leq \delta/3,\quad
\abs{f^\infty(\rho)-f^\infty(\sigma)}\leq \delta/3.
\eal
This ensures for every $\sigma\in\mB(\rho)$ that
\bal
 f^\infty(\sigma)-\frac{f_{m_\rho}(\sigma^{\otimes m_\rho})}{m_\rho}&<\abs{\frac{f_{m_\rho}(\sigma^{\otimes m_\rho})}{m_\rho}-\frac{f_{m_\rho}(\rho^{\otimes m_\rho})}{m_\rho}}+\abs{\frac{f_{m_\rho}(\rho^{\otimes m_\rho})}{m_\rho}-f^\infty(\rho)}+\abs{f^\infty(\rho)-f^\infty(\sigma)}\\
 &\leq \delta.
 \label{eq:neighborhood deviation}
\eal

Note that $\cup_{\rho\in\mS}\mB(\rho)$ is a cover of $\mS$. 
Since $\mS$ is compact, there is a finite subcover $\cup_{i=1}^K \mB(\rho_i)$ of $\mS$ for some integer $K$ with some finite set $\qty{\rho_i}_{i=1}^K$ of states in $\mS$.
Because of \eqref{eq:neighborhood deviation}, every state $\rho$ comes with some integer $i$ such that
\bal
f^\infty(\rho)-\frac{f_{m_{\rho_i}}\qty(\rho^{\otimes m_{\rho_i}})}{m_{\rho_i}}\leq \delta. 
\label{eq:continuity subcover}
\eal
Let us write an arbitrary integer $m$ as 
\bal
 m = \alpha m_{\rho_i} + \beta,\quad 0\leq \beta<m_{\rho_i}
\eal
with some nonnegative integers $\alpha$ and $\beta$ such that $\beta<m_{\rho_i}$.
Then, the superadditivity of $f_m$ ensures that 
\bal
 f_m(\rho^{\otimes m})\geq \alpha f_{m_{\rho_i}}(\rho^{\otimes m_{\rho_i}}) + f_\beta(\rho^{\otimes \beta}).
\eal
This gives
\bal
 f^\infty(\rho)-\frac{f_m(\rho^{\otimes m})}{m} &\leq f^\infty(\rho) - \frac{\alpha m_{\rho_i}}{m} \frac{f_{m_{\rho_i}}(\rho^{\otimes m_{\rho_i}})}{m_{\rho_i}} - \frac{f_\beta(\rho^{\otimes \beta})}{m}\\
 &\leq f^\infty(\rho)\qty[1-\frac{\alpha m_{\rho_i}}{m}] + \frac{\alpha m_{\rho_i}}{m}\delta  - \frac{f_\beta(\rho^{\otimes \beta})}{m}
\eal
where in the second line we used \eqref{eq:continuity subcover}. 

We now note that $\beta<m_{\rho_i}$, as well as $f_\beta(\rho^{\otimes \beta})\geq 0$ and $\sup_{\rho}f^\infty(\rho)< F$ for some constant $F$ independent of $m$.
This ensures that there exists an integer $M_i$ such that for every integer $m\geq M_i$, it satisfies  
\bal
 f^\infty(\rho)\qty[1-\frac{\alpha m_{\rho_i}}{m}] + \frac{\alpha m_{\rho_i}}{m}\delta  - \frac{f_\beta(\rho^{\otimes \beta})}{m}\leq 2\delta
\eal
for \emph{every} $\rho$ satisfying \eqref{eq:continuity subcover} for the specific state $\rho_i$.
Therefore, letting $M_{\max} \coloneqq \max_i M_i$---which is finite because $\qty{M_i}_{i=1}^K$ is a finite set---we get that for every integer $m\geq M_{\max}$, it holds that
\bal
 f^\infty(\rho) - \frac{f_m\qty(\rho^{\otimes m})}{m}\leq 2\delta
\eal
for every $\rho\in\mS$.
Since $\delta>0$ was arbitrary, we have shown the statement. 

\end{proof}

The uniform convergence of $E_d^\leq$ follows from the above two results.

\begin{cor}\label{cor:uniform continuity of distillable entanglement}
Let $\mS\subset\mbD_{>0}$ be an arbitrary compact subset of the set of full-rank distillable states. 
Then, the series $\qty{\frac{E_d^\leq(\rho_{AB}^{\otimes m})}{m}}_m$ converges to $E_d(\rho_{AB})$ uniformly on $\mS$.
\end{cor}
\begin{proof}
    The continuity (in fact, asymptotic continuity) of $E_d^\leq$ was shown in Ref.~\cite{hayashi2016mathematical} using the form \eqref{eq:another form LOCC measurement} and the asymptotic continuity of coherent information that only involves the dimension of the system $A$.
    $E_d^\leq$ is also superadditive as can easily be seen from the form \eqref{eq:another form LOCC measurement} and the fact that the tensor product of LOCC operations is also LOCC.    
    Finally, $E_d$ is continuous on $\mS$ because of \cref{lem:pointwise continuity of distillable entanglement}.
    Therefore, since $\mS$ is compact, we can apply \cref{lem:lifting pointwise to uniform continuity} to obtain the statement. 
\end{proof}

\subsubsection{Robustness of $E_d$ against noise}\label{sec:robustness}

Here, we show the results that justify our approach of ``perturb input to make it full rank and reduce the perturbation strength''.
The key property of $E_d$ toward this argument is its lower semicontinuity.

\begin{lem}\label{lem:lower semicontinous}
 The distillable entanglement $E_d:\mD(\mH)\to \mbR$ is lower semicontinuous.

\end{lem}

\begin{proof}

Eqs.~\eqref{eq:distillable entanglement two-way formal} and \eqref{eq:one-shot distillable entanglement coherent information supremum LOCC} give
\bal
 E_d(\rho_{AB}) = \lim_{n\to\infty}\sup_{\Lambda\in {\rm LOCC}}\frac{I_c(A|B)_{\Lambda(\rho_{AB}^{\otimes n})}}{n} = \sup_{n\in \mbZ_{>0},\Lambda\in{\rm LOCC}} \frac{I_c(A|B)_{\Lambda(\rho_{AB}^{\otimes n})}}{n}
\eal
where the second equality is because of Fekete's lemma, recalling that $E_d^{\leq}$ is superadditive. 
Since the supremum of continuous functions is lower semicontinuous~\cite[Proposition~1.26]{rockafellar1998variational}, and $I_c(A|B)_{\Lambda(\cdot)}:\mD(\mH^{\otimes n})\to \mbR$ is continuous for every fixed $n$ and $\Lambda$, $E_d(\cdot)$ is lower semicontinuous.
\end{proof}

The following result ensures that making the perturbation strength vanish afterward results in the correct distillable entanglement. 

\begin{lem}\label{lem:robustness}
Let $\rho_{AB,\epsilon_n}\coloneqq (1-\epsilon_n) \rho_{AB} + \epsilon_n \sigma_{AB}$ for an arbitrary state $\sigma_{AB}$ and series $\{\epsilon_n\}_n$ of positive numbers. 
Then, 
   \bal
    \lim_{n\to\infty}E_d(\rho_{AB,\epsilon_n}) = E_d(\rho_{AB})
   \eal
for an arbitrary separable state $\sigma_{AB}$ and an arbitrary series $\{\epsilon_n\}_n$ such that $\lim_{n\to\infty} \epsilon_n = 0$.   
\end{lem}

\begin{proof}

Since $E_d$ is lower semicontinuous because of \cref{lem:lower semicontinous}, we get 
\bal
\liminf_{n\to\infty} E_d(\rho_{AB,\epsilon_n})\geq E_d(\rho_{AB}).
\eal
On the other hand, since $\rho_{AB,\epsilon_n}$ can be created from $\rho_{AB}$ by LOCC, the monotonicity of $E_d$ gives 
\bal
E_d(\rho_{AB,\epsilon_n})\leq E_d(\rho_{AB})
\eal
for an arbitrary $\epsilon_n$. 
This particularly means that 
\bal
 \limsup_{n\to\infty} E_d(\rho_{AB,\epsilon_n})\leq E_d(\rho_{AB}).
\eal
Combining the above, we get 
\bal
 E_d(\rho_{AB})\leq \liminf_{n\to\infty} E_d(\rho_{AB,\epsilon_n}) \leq \limsup_{n\to\infty} E_d(\rho_{AB,\epsilon_n})\leq  E_d(\rho_{AB}),
\eal
showing the statement.
\end{proof}

\subsubsection{Protocol}\label{sec:protocol}

We now formally describe the protocol.  
Fix an arbitrary real number 
\bal
\delta>0
\label{eq:delta}
\eal
for the final target accuracy in \eqref{eq:overflow probability}.
For each integer $n$, Alice and Bob follow the strategy below.

\setcounter{protocolstep}{-1}
\stepparagraph{Prepare the initial state}
\label{para:prepare initial state}

Besides the unknown input $\rho_{AB}^{\otimes n}$, Alice and Bob prepare $\ket{0}_A\ket{0}_B$ in the output register.  This product state is retained whenever no feasible candidate is found.

\stepparagraph{Tomography to obtain a confidence region}\label{para:tomography}
Let
\bal
 \alpha_n\coloneqq \ceil{n^{1/2}},\qquad \eta_n\coloneqq n^{-2}.
 \label{eq:parameter tomography}
\eal
Alice and Bob use the first $\alpha_n$ copies for local informationally complete tomography and obtain an estimate $\widehat\rho_{AB}^{(n)}$.  
Let 
\bal
G_n(\rho_{AB})\coloneqq \lset \widehat\rho_{AB} \sbar\norm{\widehat\rho_{AB}-\rho_{AB}}_1\leq \epsilon_n\rset
 \label{eq:common tomography good event}
\eal
be the set of good estimates when the true state is $\rho_{AB}$.
Then, it is well known that using $\alpha_n$ samples ensures that the estimate $\widehat\rho_{AB}^{(n)}$ is contained in $G_n(\rho_{AB})$ with probability at least $1-\eta_n$ where
\bal
\epsilon_n=O\qty(\sqrt{\frac{\log(1/\eta_n)}{\alpha_n}})=O\qty(n^{-1/4}\sqrt{\log n}).
 \label{eq:common tomography radius}
\eal

\stepparagraph{Deciding perturbation strength, block size, and LOCC measurement}\label{para:deciding parameters}

In the later steps, Alice and Bob perturb the input with some noise strength (\ref{para:perturb}), make an LOCC measurement (\ref{para:LOCC}), and apply state merging (\ref{para:state merging}). 
In the current step, they decide the parameters used in the following steps. 
They do so by identifying the set of parameters by evaluating the performance of the final state merging step realized by each of the potential parameters, then choosing the one among them so that the desired asymptotic performance is guaranteed. 
Therefore, this step already contains the analysis of the protocol that is run in the later steps. 

In \ref{para:perturb}, in order for us to employ \cref{cor:uniform continuity of distillable entanglement}, we consider perturbing the input to make it full rank. 
Let $\{\xi_k\}_k$ be a decreasing sequence of real numbers such that $\xi_k\xrightarrow[k\to\infty]{} 0$. 
For instance, we take $\xi_k=1/k$. 
Let
\bal
\rho_{AB,\xi_k}\coloneqq (1-\xi_k)\rho_{AB}+\xi_k\frac{\mdI}{d_Ad_B}
\eal
be the perturbed input states corresponding to each noise strength $\xi_k$.
We also consider the perturbed estimate
\bal
\widehat\rho_{AB,\xi_k}^{(n)}\coloneqq (1-\xi_k)\widehat\rho_{AB}^{(n)}+\xi_k\frac{\mdI}{d_Ad_B}.
\eal
If $\rho_{AB}\in G_n$ defined in \eqref{eq:common tomography good event}, it is ensured that 
\bal
\norm{\widehat\rho_{AB,\xi_k}^{(n)}-\rho_{AB,\xi_k}}_1
 \leq (1-\xi_k)\epsilon_n
 \label{eq:perturbed common radius}
\eal
for every $k$.
For each candidate $k$, define the closed confidence set
\bal
 \mS_{n,k}\coloneqq \lset\sigma_{AB}\in\mD(\mH_{AB})\sbar \norm{\sigma_{AB}-\widehat\rho_{AB,\xi_k}^{(n)}}_1\leq (1-\xi_k)\epsilon_n\rset.
 \label{eq:candidate confidence set}
\eal

Alice and Bob choose an appropriate $k\in \qty{1,\ldots, n}$ as follows.
The overall idea is that they would like to choose as large $k$ as possible, as a larger $k$ corresponds to a smaller noise strength, which would make the perturbed state close to the true state. 
At least, they would like to make sure that their way of choosing $k$ grows with $n$ so that in the limit of $n\to\infty$, it also ensures that $k\to\infty$. 
However, if they take $k$ that grows with $n$ too fast, the perturbed state can approach the boundary too fast, which may result in the case where the confidence region $\mS_{n,k}$ is not contained in the set of full-rank distillable states and does not allow them to choose an appropriate block size $m$ due to \cref{cor:uniform continuity of distillable entanglement}.

In addition, the strategy of choosing $k$ needs to make sure that, not only is the corresponding confidence region included in the interior of the set of distillable states, but one can also take an appropriate block size $m$ for state merging and a preprocessing LOCC measurement $\Lambda_{\widehat \mM}$ such that the target rate and fidelity are guaranteed to be achieved in the $n\to\infty$ limit.

We address this issue by not predetermining the scaling of $k$ with $n$, but by checking whether each $k\in \qty{1,\ldots, n}$ is feasible, and choosing the maximum one among them.
Here, Alice and Bob examine each $k$ by checking the above criteria by classical computation, i.e., whether the confidence region is included in the interior of the set of distillable states, and whether the target rate and infidelity can be achieved after state merging.

Here is the concrete protocol that Alice and Bob take. 
For each $k$, they first see whether 
\bal
 \mS_{n,k}\subset\mbD_{>0}.
 \label{eq:candidate inside D}
\eal
where $\mbD_{>0}$ is the set of full-rank distillable states defined in \cref{cor:uniform continuity of distillable entanglement}.
They reject this candidate $k$ if this is not satisfied. 

If \eqref{eq:candidate inside D} holds, $\mS_{n,k}$ is compact and \cref{cor:uniform continuity of distillable entanglement} applies.
They see whether the size $(1-\xi_k)\epsilon_n$ of the confidence region given in \eqref{eq:perturbed common radius} is small enough to guarantee the desired achievable rate. 
Specifically, for $m\in\mbZ_{>0}$ and $t\geq0$, define
\bal
 \omega_m(t)\coloneqq 2mt\log d_A +(1+mt)h\qty(\frac{mt}{1+mt}).
 \label{eq:Ic continuity modulus}
\eal
Then, the continuity of coherent information in \eqref{eq:continuity coherent information} ensures that, for every state $\sigma_{AB}$, $\widehat\sigma_{AB}$, and LOCC measurement channel $\Lambda_\mM$,
\bal
\abs{\frac{I_c(A|BE)_{\Lambda_{\mM}(\sigma_{AB}^{\otimes m})}}{m}-\frac{I_c(A|BE)_{\Lambda_{\mM}(\widehat\sigma_{AB}^{\otimes m})}}{m}}
 \leq \omega_m\qty(\frac{1}{2}\norm{\sigma_{AB}-\widehat\sigma_{AB}}_1)
 \label{eq:Ic common continuity}
\eal
where we assume that $\frac{m}{2}\norm{\sigma_{AB}-\widehat\sigma_{AB}}_1\leq 1$.

We use the continuity bound \eqref{eq:Ic common continuity} to guarantee that the coherent information of the estimated perturbed state is close to that of the true perturbed state.
But, we also observe that the continuity bound \eqref{eq:Ic common continuity} depends on the choice of $m$, which we choose so that it approximates the distillable entanglement by the coherent-information-based quantity $E_d^{\leq}/m$ well enough based on the uniform convergence of given in \cref{cor:uniform continuity of distillable entanglement}.
Errors in these two---the deviation of the coherent information of the estimated perturbed state from that of the true perturbed state, and the deviation of $E_d^\leq/m$ from $E_d$---should be kept small simultaneously. 
If that is impossible, Alice and Bob reject the current candidate $k$.

Alice and Bob accomplish this check as follows. 
We first introduce a positive constant $\delta_0>0$ given by 
\bal
 \delta_0 \coloneqq \frac{\delta}{8},
 \label{eq:delta_0}
\eal
which we choose in this way so that the final target accuracy $\delta$ can be achieved in the end.
Alice and Bob begin by choosing an integer $m$ such that the deviation of $E_d^\leq/m$ from $E_d$ is less than $\delta_0$ for every state in $\mS_{n,k}$. 
We write such $m$ as $m_{\delta_0}(\mS_{n,k})$ where we define $m_\zeta(\mS)$ for arbitrary $\zeta>0$ and compact $\mS\subset\mbD_{>0}$ by
\bal
 m_\zeta(\mS)\coloneqq \min\lset M\in\mbZ_{>0}\sbar E_d(\sigma_{AB})-\frac{E_d^\leq(\sigma_{AB}^{\otimes m})}{m}<\zeta\ \text{for every $\sigma_{AB}\in\mS$ and every $m\geq M$}\rset.
 \label{eq:uniform convergence threshold}
\eal
\cref{cor:uniform continuity of distillable entanglement} guarantees that $m_\zeta(\mS)<\infty$.
Alice and Bob check whether this choice of $m=m_{\delta_0}(\mS_{n,k})$ is small enough so that the right-hand side of \eqref{eq:Ic common continuity} can be upper bounded by $\delta_0$ for two states in $\mS_{n,k}$---one we think of as the true perturbed state and the other the estimated perturbed state.
Namely, Alice and Bob reject $k$ unless
\bal
 \omega_{m_{\delta_0}(\mS_{n,k})}((1-\xi_k)\epsilon_n)\leq\delta_0.
 \label{eq:tomography accuracy feasibility}
\eal

For a candidate that has survived the above tests \eqref{eq:candidate inside D} and \eqref{eq:tomography accuracy feasibility}, Alice and Bob use the classical description of $\widehat\rho_{AB,\xi_k}^{(n)}$ to choose a finite-output LOCC measurement channel $\Lambda_{\mM_{n,k}}$ such that
\bal
\frac{I_c(A|BE)_{\Lambda_{\mM_{n,k}}\qty[(\widehat\rho_{AB,\xi_k}^{(n)})^{\otimes m_{\delta_0}(\mS_{n,k})}]}}{m_{\delta_0}(\mS_{n,k})}
 \geq \frac{E_d^\leq\qty[(\widehat\rho_{AB,\xi_k}^{(n)})^{\otimes m_{\delta_0}(\mS_{n,k})}]}{m_{\delta_0}(\mS_{n,k})}-\delta_0.
 \label{eq:near optimal candidate measurement}
\eal
Among all finite-output LOCC POVMs satisfying \eqref{eq:near optimal candidate measurement}, they choose the one with the smallest number of outcomes.
We write such smallest number of outcomes as $q_{n,k}$.  

They now investigate whether the state merging in \ref{para:state merging} with this choice of parameters achieves the target rate and fidelity.
Specifically, they first check whether 
\bal
\frac{I_c(A|BE)_{\Lambda_{\mM_{n,k}}\qty[(\widehat\rho_{AB,\xi_k}^{(n)})^{\otimes m_{\delta_0}(\mS_{n,k})}]}}{m_{\delta_0}(\mS_{n,k})}-\delta_0-\frac{\delta_0}{m_{\delta_0}(\mS_{n,k})}>0,
\label{eq:rate test}
\eal
where the second term is due to \eqref{eq:tomography accuracy feasibility} and the third term is the rate loss corresponding to $\delta$ in \cref{lem:universal state merging}, as we see later.
If this is not satisfied, they reject this choice of $k$.

If it is satisfied, they consider applying the state merging to $l_{n,k}$ blocks of the state $\Lambda_{\mM_{n,k}}(\sigma^{\otimes m_{\delta_0}(\mS_{n,k})})$ with $\sigma\in\mS_{n,k}$, where 
\bal
 l_{n,k}\coloneqq \floor{\frac{n-\alpha_n}{m_{\delta_0}(\mS_{n,k})}},
\label{eq:candidate number blocks}
\eal
and $\alpha_n=o(n)$ is the number of copies used for the state tomography in \ref{para:tomography}.
They check whether 
\bal
l_{n,k}>n_0(\delta_0,d_{n,k})
 \label{eq:block size condition}
\eal
where $n_0$ is the function in \cref{lem:universal state merging}, and 
\bal
d_{n,k}\coloneqq d_A^{m_{\delta_0}(\mS_{n,k})}d_B^{m_{\delta_0}(\mS_{n,k})}q_{n,k}
\label{eq:output dimension}
\eal
is the dimension of $\Lambda_{\mM_{n,k}}\qty[(\widehat\rho_{AB,\xi_k}^{(n)})^{\otimes m_{\delta_0}(\mS_{n,k})}]$, recalling that $q_{n,k}$ is the number of measurement outcomes of $\mM_{n,k}$ defined below \eqref{eq:near optimal candidate measurement}.
If it is not satisfied, they reject this choice of $k$.

When it is passed, \cref{lem:universal state merging} and \eqref{eq:tomography accuracy feasibility} ensure that the rate
\bal
\frac{\log d_{\rm fin}}{n}&\geq   \qty(\frac{I_c(A|BE)_{\Lambda_{\mM_{n,k}}\qty[(\widehat\rho_{AB,\xi_k}^{(n)})^{\otimes m_{\delta_0}(\mS_{n,k})}]}}{m_{\delta_0}(\mS_{n,k})}-\delta_0-\frac{\delta_0}{m_{\delta_0}(\mS_{n,k})})\frac{m_{\delta_0}(\mS_{n,k})l_{n,k}}{n}
\label{eq:achievable rate}
\eal
can be achieved with trace-distance error $2^{-c(\delta_0,d_{n,k})l_{n,k}}$, where $c$ is the function in \cref{lem:universal state merging}. 
Setting our target merging error $\varepsilon_n = 1/\poly(n)$ (for instance, $n^{-1}$), they check whether
\bal
 2^{-c(\delta_0,d_{n,k})l_{n,k}}< \varepsilon_n
\label{eq:fidelity test}
\eal
holds. 
If this is not satisfied, they reject this choice of $k$.

If $k$ passed all tests \eqref{eq:candidate inside D}, \eqref{eq:tomography accuracy feasibility}, \eqref{eq:rate test}, \eqref{eq:block size condition}, and \eqref{eq:fidelity test}, we call $k$ feasible.
Let $K_n^{\rm feas}\subset [0,n]$ be the set of feasible $k$'s. 
Alice and Bob choose
\bal
 k_n\coloneqq \max K_n^{\rm feas}.
 \label{eq:selected k}
\eal
If $K_n^{\rm feas}=\varnothing$, Alice and Bob keep the initial product output and terminate.  Importantly, the entire search in \ref{para:deciding parameters} is classical; rejecting a candidate $k$ does not consume any additional quantum copies.

\stepparagraph{Perturb input states}\label{para:perturb}
Suppose $K_n^{\rm feas}\neq\varnothing$.  
Alice and Bob apply the depolarizing noise to prepare $\rho_{AB,\xi_{k_n}}$ from $\rho_{AB}$ to each of the remaining $n-\alpha_n$ copies.  
Such depolarizing noise can be applied with shared randomness, which can be sent from Alice to Bob, and local depolarizing noise applied depending on the shared random bit.

\stepparagraph{Apply LOCC measurement}\label{para:LOCC}
They divide the resulting $n-\alpha_n$ perturbed copies into $l_{n,k_n}$ blocks of size $m_{\delta_0}(\mS_{n,k_n})$, where $l_{n,k_n}$ is given in \eqref{eq:candidate number blocks} with the $k_n$ chosen as in \eqref{eq:selected k}.
If $n-\alpha_n$ is not divisible by $m_{\delta_0}(\mS_{n,k_n})$, they discard the leftover copies, which is fewer than $m_{\delta_0}(\mS_{n,k_n})$.

They apply the LOCC measurement channel $\Lambda_{\mM_{n,k_n}}$ found in \ref{para:deciding parameters} to every block independently.
Since the perturbed input belongs to $\mS_{n,k_n}$, each block output belongs to the known set
\bal
 \mT_n\coloneqq \lset \Lambda_{\mM_{n,k_n}}(\sigma_{AB}^{\otimes m_{\delta_0}(\mS_{n,k_n})})\sbar \sigma_{AB}\in\mS_{n,k_n}\rset.
 \label{eq:state merging candidate set}
\eal
By \eqref{eq:Ic common continuity}, we get
\bal
 \min_{\tau\in\mT_n} \frac{I_c(A|BE)_\tau}{m_{\delta_0}(\mS_{n,k_n})}\geq  \frac{I_c(A|BE)_{\Lambda_{\mM_{n,k_n}}\qty[(\widehat\rho_{AB,\xi_{k_n}}^{(n)})^{\otimes m_{\delta_0}(\mS_{n,k_n})}]}}{m_{\delta_0}(\mS_{n,k_n})}-\delta_0.
 \label{eq:certified minimum Ic}
\eal

\stepparagraph{Apply universal state merging}\label{para:state merging}
As already described in \ref{para:deciding parameters}, Alice and Bob apply the state-merging protocol with partial information in \cref{lem:universal state merging} to the $l_{n,k_n}$ i.i.d. block outputs, using the known uncertainty set $\mT_n$. 
Specifically, they run the state merging by setting the target rate $r_n$ as the achievable one given in \eqref{eq:achievable rate} with $k=k_n$, i.e., 
\bal
 r_n\coloneqq \qty(\frac{I_c(A|BE)_{\Lambda_{\mM_{n,k_n}}\qty[(\widehat\rho_{AB,\xi_{k_n}}^{(n)})^{\otimes m_{\delta_0}(\mS_{n,{k_n}})}]}}{m_{\delta_0}(\mS_{n,k_n})}-\delta_0-\frac{\delta_0}{m_{\delta_0}(\mS_{n,k_n})})\frac{m_{\delta_0}(\mS_{n,k_n})l_{n,k_n}}{n}.
 \label{eq:revised target rate}
\eal
The trace-norm error of this state-merging step is at most $\varepsilon_n=1/\poly(n)$ by \eqref{eq:fidelity test}.

\subsubsection{Analysis} \label{sec:analysis}

We now analyze the above protocol.  
We first show that the maximum feasible $k$ determined at every $n$ is guaranteed to diverge in the asymptotic limit.
\begin{pro}\label{pro:maximum feasible diverges}
 Suppose that $E_d(\rho_{AB})>\delta$.
 For a series of good estimates $\qty{\widehat\rho_{AB}^{(n)}}_n$, $k_n$ defined in \eqref{eq:selected k} satisfies $k_n\to\infty$ as $n\to\infty$.  
\end{pro}
\begin{proof}
    We first show that, for a sufficiently large but fixed integer $k$, it eventually becomes feasible in the limit of $n\to\infty$. 
We then employ this fact to show the desired statement.

Let $k$ be a sufficiently large but fixed integer.
How large this $k$ needs to be can depend on the true state $\rho_{AB}$, but it does not affect the feasibility of the universal protocol---the existence of such $k$ suffices for the argument. 
There are five tests \eqref{eq:candidate inside D}, \eqref{eq:tomography accuracy feasibility}, \eqref{eq:rate test}, \eqref{eq:block size condition}, and \eqref{eq:fidelity test} to pass for the feasibility, and we show the eventual feasibility of each test in the following.

\paragraph*{Eventual feasibility of \eqref{eq:candidate inside D}}
It is easy to see that \eqref{eq:candidate inside D} is eventually satisfied for a sufficiently large $n$ because the tomography accuracy keeps increasing with $n$. 

\paragraph*{Eventual feasibility of \eqref{eq:tomography accuracy feasibility}}
The nontrivial part of this condition is the behavior of $m_{\delta_0}(\mS_{n,k})$ with growing $n$---although $\epsilon_n$ decays as in \eqref{eq:common tomography radius} with $n$, \eqref{eq:tomography accuracy feasibility} may not be passed if $m_{\delta_0}(\mS_{n,k})$ also grows with $n$.
We address this by showing that $m_{\delta_0}(\mS_{n,k})$ can be upper bounded by a quantity independent of $n$, ensuring that it does not diverge in $n\to\infty$ limit. 
To see this, we first notice that \cref{lem:robustness} ensures $E_d(\rho_{AB,\xi_k})\to E_d(\rho_{AB})$ as $k\to\infty$, and thus every sufficiently large fixed $k$ satisfies $\rho_{AB,\xi_k}\in\mbD_{>0}$, where $\mbD_{>0}$ is the set of full-rank distillable states.  Since $\mbD_{>0}$ is open, there is a compact set $\mS_k\subset\mbD_{>0}$ whose interior contains $\rho_{AB,\xi_k}$.  
Again, Alice and Bob may not know the identity of $\mS_k$, but that is not a problem---we employ this for the sake of the proof. 
If $\widehat\rho_{AB}^{(n)}\in G_n(\rho_{AB})$, the center of $\mS_{n,k}$ is within $(1-\xi_k)\epsilon_n$ of $\rho_{AB,\xi_k}$ and the radius is $(1-\xi_k)\epsilon_n\xrightarrow[n\to\infty]{}0$.
Therefore,
\bal
 \mS_{n,k}\subset\mS_k
 \label{eq:candidate contained fixed compact}
\eal
for every sufficiently large $n$. 
We now realize that the definition \eqref{eq:uniform convergence threshold} ensures that 
\bal
\mS'\subset\mS\quad\Longrightarrow\quad m_\zeta(\mS')\leq m_\zeta(\mS).
 \label{eq:uniform convergence threshold monotonicity}
\eal
Consequently, we get
\bal
 m_{\delta_0}(\mS_{n,k})\leq m_{\delta_0}(\mS_k)<\infty
 \label{eq:fixed k block bound}
\eal
for all sufficiently large $n$.  
Since $m_{\delta_0}(\mS_{n,k})$ is upper bounded by constant, $\omega_{m_{\delta_0}(\mS_{n,k})}(2(1-\xi_k)\epsilon_n)$ can be made arbitrarily small by taking sufficiently large $n$, making $(1-\xi_k)\epsilon_n$ small. 
This ensures that \eqref{eq:tomography accuracy feasibility} is eventually satisfied. 

\paragraph*{Eventual feasibility of \eqref{eq:rate test}}
To see that the rate test \eqref{eq:rate test} is also passed for a fixed sufficiently large $k$ and in the limit of $n\to\infty$, it suffices to focus on the $\rho_{AB}$ such that $E_d(\rho_{AB})>\delta$.
Indeed, if $E_d(\rho_{AB})\leq \delta$, a product output satisfies both conditions \eqref{eq:error condition} and \eqref{eq:overflow probability} for universal entanglement distillation, and thus the failure of passing the test does not affect the feasibility of the protocol. 
If $E_d(\rho_{AB})>\delta$, since $E_d(\rho_{AB,\xi_k})\xrightarrow[k\to\infty]{} E_d(\rho_{AB})$ as ensured in \cref{lem:robustness}, a sufficiently large $k$ satisfies $E_d(\rho_{AB,\xi_k})>\delta$. 
Using \eqref{eq:uniform convergence threshold} and \eqref{eq:near optimal candidate measurement}, the choice of $\epsilon_n$ in \eqref{eq:common tomography radius} that decays with $n$, and that $m_{\delta_0}(\mS_{n,k})$ is upper bounded by a constant as in \eqref{eq:fixed k block bound} ensures that \eqref{eq:rate test} is passed for a sufficiently large $n$ for the fixed $k$.

\paragraph*{Eventual feasibility of \eqref{eq:fidelity test}}
Before going to \eqref{eq:block size condition}, we first show that \eqref{eq:fidelity test} is eventually satisfied. 
To see this, we first note that \eqref{eq:fixed k block bound} and the definition of $l_{n,k}$ in \eqref{eq:candidate number blocks} ensure that $l_{n,k}=\Omega(n)$.
Recalling that $\delta_0$ is a constant and $c(\delta_0,d_{n,k})$ in \eqref{eq:fidelity test} is non-increasing in $d_{n,k}$, it suffices to show that the output dimension $d_{n,k}=d_A^{m_{\delta_0}(\mS_{n,k})}d_B^{m_{\delta_0}(\mS_{n,k})}q_{n,k}$ defined in \eqref{eq:output dimension} is upper bounded by a constant with $n$. 
And since we already saw that $m_{\delta_0}(\mS_{n,k})$ is upper bounded by a constant, we can focus on $q_{n,k}$---the number of outcomes in the LOCC POVM.

To proceed, consider an arbitrary integer $m$ such that $m\in\{1,\ldots,m_{\delta_0}(\mS_k)\}$.
For each such $m$, choose a finite-output LOCC POVM $\mM_{k,m}^\star$ satisfying
\bal
\frac{I_c(A|BE)_{\Lambda_{\mM_{k,m}^\star}(\rho_{AB,\xi_k}^{\otimes m})}}{m}
 \geq \frac{E_d^\leq(\rho_{AB,\xi_k}^{\otimes m})}{m}-\frac{\delta_0}{2}.
 \label{eq:fixed k comparison measurement}
\eal
For fixed $m$, both sides of \eqref{eq:fixed k comparison measurement} are continuous in the state $\rho_{AB,\xi_k}$.   
Hence, whenever $m_{\delta_0}(\mS_{n,k})=m$, the same $\mM_{k,m}^\star$ satisfies \eqref{eq:near optimal candidate measurement} for every sufficiently large $n$.  
Since $\mM_{n,k}$ is chosen with the smallest outcome cardinality among the measurements satisfying \eqref{eq:near optimal candidate measurement}, there is a finite constant $q_k$, which can be taken as the maximum number of measurement outcomes over all $m_{\delta_0}(\mS_{n,k})\in\{1,\dots,m_{\delta_0}(\mS_k)\}$, such that
\bal
 q_{n,k}\leq q_k <\infty
 \label{eq:fixed k outcome bound}
\eal
for every sufficiently large $n$.
This shows that \eqref{eq:fidelity test} is satisfied for a sufficiently large $n$.

\paragraph*{Eventual feasibility of \eqref{eq:block size condition}}
We can now see that \eqref{eq:block size condition} can be satisfied too. 
This is because $l_{n,k}=\Omega(n)$ and $d_{n,k}$ is upper bounded by a constant. 
Since $n_0$ is non-decreasing with $d_{n,k}$, there is a sufficiently large $n$ such that $l_{n,k}>n_0(\delta_0,d_{n,k})$.

\vspace{.5cm}

The above observations now allow us to see that $k_n\to \infty$ as $n\to\infty$.
Let us fix an arbitrary sufficiently large integer $K$. 
The preceding argument shows that there exists $N_K$ such that $K$ is feasible for every $n\geq N_K$ whenever the estimate belongs to $G_n(\rho_{AB})$. Since $k_n$ is the largest feasible candidate, it follows that $k_n\geq K$ for every $n\geq N_K$. 
As $K$ is arbitrary, we get that $k_n\to\infty$. 
\end{proof}

\cref{pro:maximum feasible diverges} ensures that we can keep decreasing the perturbation strength while ensuring the target rate and error of the entanglement distillation. 
This, together with the analysis of the performance in \ref{para:state merging}, results in the following result, which concludes the proof of \cref{thm:universal distillable entanglement main}.

\begin{pro}\label{pro:analysis of protocol}
The protocol in Sec.~\ref{sec:protocol} achieves the distillable entanglement universally in the sense of \cref{defn:universal distillable entanglement}.
\end{pro}

\begin{proof}
Let $k_n\coloneqq\max K_n^{\rm feas}$ as in \eqref{eq:selected k}.  
By \eqref{eq:uniform convergence threshold}, it holds that
\bal
 0\leq E_d(\sigma_{AB})-\frac{E_d^\leq(\sigma_{AB}^{\otimes m_{\delta_0}(\mS_{n,k_n})})}{m_{\delta_0}(\mS_{n,k_n})}<\delta_0,
 \qquad \forall \sigma_{AB}\in\mS_{n,k_n}.
 \label{eq:uniform gap candidate}
\eal
Moreover, \eqref{eq:tomography accuracy feasibility} and \eqref{eq:Ic common continuity} imply that, for every LOCC measurement channel $\Lambda_\mM$ and every $\sigma_{AB}^{(1)},\sigma_{AB}^{(2)}\in\mS_{n,k_n}$,
\bal
 \abs{\frac{I_c(A|BE)_{\Lambda_\mM\qty((\sigma_{AB}^{(1)})^{\otimes m_{\delta_0}(\mS_{n,k_n})})}}{m_{\delta_0}(\mS_{n,k_n})}-\frac{I_c(A|BE)_{\Lambda_\mM\qty((\sigma_{AB}^{(2)})^{\otimes m_{\delta_0}(\mS_{n,k_n})})}}{m_{\delta_0}(\mS_{n,k_n})}}\leq\delta_0.
 \label{eq:pairwise Ic candidate}
\eal

For arbitrary $\sigma_{AB}^{(1)},\sigma_{AB}^{(2)}\in\mS_{n,k_n}$, choose a finite-outcome LOCC POVM $\mM^{(2)}$ that is $\delta_0$-optimal for $E_d^\leq((\sigma_{AB}^{(2)})^{\otimes m_{\delta_0}(\mS_{n,k_n})})$.  Then
\bal E_d(\sigma_{AB}^{(1)})
 &\geq \frac{E_d^\leq((\sigma_{AB}^{(1)})^{\otimes m_{\delta_0}(\mS_{n,k_n})})}{m_{\delta_0}(\mS_{n,k_n})} \\
 &\geq \frac{I_c(A|BE)_{\Lambda_{\mM^{(2)}}\qty[(\sigma_{AB}^{(1)})^{\otimes m_{\delta_0}(\mS_{n,k_n})}]}}{m_{\delta_0}(\mS_{n,k_n})} \\
 &\geq \frac{I_c(A|BE)_{\Lambda_{\mM^{(2)}}\qty[(\sigma_{AB}^{(2)})^{\otimes m_{\delta_0}(\mS_{n,k_n})}]}}{m_{\delta_0}(\mS_{n,k_n})}-\delta_0 \\
 &\geq E_d(\sigma_{AB}^{(2)})-3\delta_0,
 \label{eq:Ed oscillation one direction}
\eal
where the first line follows from \eqref{eq:uniform gap candidate}, the second line is because $E_d^\leq$ is obtained by maximizing over LOCC measurements as in \eqref{eq:another form LOCC measurement}, the third line is due to \eqref{eq:pairwise Ic candidate}, and the fourth line comes from that $\mM^{(2)}$ is $\delta_0$-optimal for $E_d^\leq\qty((\sigma_{AB}^{(2)})^{\otimes m_{\delta_0}(\mS_{n,k_n})})$ together with \eqref{eq:uniform gap candidate}.
Also, interchanging $\sigma_{AB}^{(1)}$ and $\sigma_{AB}^{(2)}$ yields
\bal
 \abs{E_d(\sigma_{AB}^{(1)})-E_d(\sigma_{AB}^{(2)})}\leq3\delta_0,
 \qquad \forall \sigma_{AB}^{(1)},\sigma_{AB}^{(2)}\in\mS_{n,k_n},
 \label{eq:Ed oscillation candidate}
\eal

The choice \eqref{eq:near optimal candidate measurement} and \eqref{eq:uniform gap candidate} imply
\bal
\frac{I_c(A|BE)_{\Lambda_{\mM_{n,k_n}}}\qty(\qty[\widehat\rho_{AB,\xi_k}^{(n)}]^{\otimes m_{\delta_0}(\mS_{n,k_n})})}{m_{\delta_0}(\mS_{n,k_n})}
 \geq E_d(\widehat{\rho}_{AB,\xi_k}^{(n)})-2\delta_0.
\eal
Therefore, the rate $r_n$ in \eqref{eq:revised target rate} achieved in \ref{para:state merging} satisfies
\bal
r_n \geq \qty(E_d(\widehat{\rho}_{AB,\xi_{k_n}}^{(n)})-3\delta_0-\frac{\delta_0}{m_{\delta_0}(\mS_{n,k_n})})\frac{l_{n,k_n}m_{\delta_0}(\mS_{n,k_n})}{n}.
 \label{eq:L versus estimate}
\eal
If $\widehat{\rho}_{AB}^{(n)}\in G_n(\rho_{AB})$, Eq.~\eqref{eq:perturbed common radius} implies $\rho_{AB,\xi_k}\in\mS_{n,k_n}$.
Noting also that $\widehat{\rho}_{AB,\xi_k}^{(n)}\in\mS_{n,k_n}$, \eqref{eq:Ed oscillation candidate}, together with \eqref{eq:L versus estimate}, gives
\bal
 r_n \geq \qty(E_d(\rho_{AB,\xi_{k_n}})-6\delta_0-\frac{\delta_0}{m_{\delta_0}(\mS_{n,k_n})})\frac{l_{n,k_n}m_{\delta_0}(\mS_{n,k_n})}{n}.
 \label{eq:L versus regularized true state}
\eal

We now realize that 
\bal
 \frac{l_{n,k_n}m_{\delta_0}(\mS_{n,k_n})}{n}\geq 1 - \frac{\alpha_n}{n} - \frac{m_{\delta_0}(\mS_{n,k_n})}{n} \xrightarrow[n\to\infty]{} 1 .
\eal
Here, the second term vanishes because $\alpha_n=o(n)$.
The third term also vanishes as 
\bal
 \frac{m_{\delta_0}(\mS_{n,k_n})}{n}\leq \frac{1}{l_{n,k}}\leq \frac{c(\delta_0,d_{n,k_n})}{\log(1/\varepsilon_n)}\leq \frac{c(\delta_0,d_Ad_B)}{\log(1/\varepsilon_n)}\xrightarrow[n\to\infty]{}0
\eal
where the second inequality follows from \eqref{eq:fidelity test}, the third inequality is because $c$ is non-increasing in $d_{n,k}$ and $d_{n,k}\geq d_Ad_B$, and the limit follows because $\varepsilon_n=1/\poly(n)$.
Also, since $m_{\delta_0}(\mS_{n,k_n})$ is a positive integer, we have $\frac{\delta_0}{m_{\delta_0}(\mS_{n,k_n})}\leq \delta_0$.

Finally, \cref{pro:maximum feasible diverges} ensures $\xi_{k_n}\xrightarrow[n\to\infty]{}0$.
Therefore, \cref{lem:robustness} implies
\bal
E_d(\rho_{AB,\xi_{k_n}})\xrightarrow[n\to\infty]{} E_d(\rho_{AB}).
 \label{eq:random regularization removal}
\eal
These show that the achievable rate $r_n$ for a good estimate in $G_n(\rho_{AB})$ satisfies 
\bal
 \liminf_{n\to\infty}r_n \geq E_d(\rho_{AB})-7\delta_0.
 \label{eq:achievable rate limit}
\eal

We are now ready to verify the conditions \eqref{eq:error condition} and \eqref{eq:overflow probability}. 
Eq.~\eqref{eq:overflow probability} can be seen by \eqref{eq:achievable rate limit} together with that the probability of having a good estimate in $G_n(\rho_{AB})$ approaches 1 because of the choice of $\eta_n$ in \eqref{eq:parameter tomography}. 
Specifically the choice of $\delta_0$ in \eqref{eq:delta_0} gives
\bal
 \Pr_x\qty[r_x^{(n)}<E_d(\rho_{AB})-\delta]\xrightarrow[n\to\infty]{}0.
 \label{eq:revised overflow proof}
\eal
To see \eqref{eq:error condition}, we note that if $\widehat\rho_{AB}^{(n)}\in G_n(\rho_{AB})$, its state-merging error is at most $\varepsilon_n=1/\poly(n)$ by construction.  
If no candidate is selected, the output rate is defined to be zero and the protocol retains the target product state exactly.  
If $\widehat\rho_{AB}^{(n)}\not\in G_n(\rho_{AB})$, the trace distance is at most $2$, and this only occurs with probability at most $\eta_n=n^{-2}$.  
Hence,
\bal \varepsilon_n(\rho_{AB})\leq \varepsilon_n +2\eta_n\xrightarrow[n\to\infty]{}0,
 \label{eq:revised average error}
\eal
which proves \eqref{eq:error condition}.

Finally, if $E_d(\rho_{AB})\leq \delta$, the condition \eqref{eq:overflow probability} holds automatically. 
The condition \eqref{eq:error condition} follows from \eqref{eq:revised average error}.
\end{proof}

\subsubsection{Remark on one-way distillable entanglement}

We remark that, if we replace the two-way LOCC measurement $\Lambda_\mM$ with one-way LOCC measurement, our protocol achieves the one-way distillable entanglement $E_{d,\rightarrow}(\rho_{AB})$ universally, because the one-way distillable entanglement can also be characterized as in \eqref{eq:distillable entanglement two-way formal} and \eqref{eq:another form LOCC measurement} by replacing two-way LOCC with one-way LOCC. 
Strictly speaking, this universal protocol itself is a two-way LOCC protocol because Bob needs to send the tomographic data back to Alice in \ref{para:tomography} so that they agree on the common state estimate. 
However, the rate of this classical communication vanishes in the asymptotic limit because of the same reasoning in \cref{thm:universal entanglement distillation hashing}.
Therefore, our protocol achieves the distillable entanglement universally in an ``almost'' one-way LOCC.

\end{document}